\documentclass[11pt]{article}

\usepackage{amssymb,amsmath,amsthm}
\usepackage{bbm}
\usepackage[noend]{algpseudocode}
\usepackage{graphicx}
\usepackage{verbatim}
\usepackage{url,xspace}
\usepackage[colorlinks=true,linkcolor=blue]{hyperref}
\usepackage{cite}
\usepackage{caption}
\usepackage{subcaption}
\usepackage{multirow}
\usepackage{multicol}
\usepackage{latexsym}
\usepackage{amsmath,amssymb,enumerate}
\usepackage{xspace}
\usepackage{algorithm,algorithmicx}
\usepackage{float}
\usepackage{xcolor}
\usepackage{mathrsfs}
\usepackage{mathtools}
\usepackage{relsize}
\usepackage{booktabs}

\usepackage[margin=1in]{geometry}

\newtheorem{theorem}{Theorem}[section]
\newtheorem{definition}{Definition}[section]
\newtheorem{proposition}{Proposition}

\newtheorem{lemma}[theorem]{Lemma}

\newtheorem{corollary}[theorem]{Corollary}

\newcounter{note}[section]

\def\sse{\subseteq}

\newcommand{\pr}{\mathbf{P}} 
\newcommand{\E}{\mathbb{E}} 
\newcommand{\ZZ}{\mathbb{Z}}

\newcommand{\R}{\mathbb{R}}

\newcommand{\D}{\mathcal{D}}

\newcommand{\OPT}{\mathtt{OPT}}
\newcommand{\NA}{\mathtt{NA}}

\newcommand{\Stem}{\mathtt{Stem}}

\newcommand{\scn}{\omega}
\newcommand{\score}{\text{score}}

\newcommand{\bigOh}{\mathcal{O}}

\newcommand{\ignore}[1]{}

\newcommand{\prob}{$\mathtt{IPP}$\xspace}
\newcommand{\probh}{$\mathtt{IPP-H}$\xspace}

\newcommand{\RSO}{$\mathtt{RSO}$\xspace}

\def\H{{\cal H}}

\title{Informative Path Planning with Limited Adaptivity}
\author{Rayen Tan\thanks{Department of Industrial and Operations Engineering, University of Michigan, Ann Arbor, USA. Research supported in part by NSF grants CMMI-1940766 and CCF-2006778.} \and Rohan Ghuge\thanks{Department of Industrial and Systems Engineering / Algorithms and Randomness Center, Georgia Institute of Technology, Atlanta, USA.} \and Viswanath Nagarajan$^*$}
\date{October 2023}

\begin{document}

\maketitle
\begin{abstract}
We consider the \emph{informative path planning} (\prob) problem 
in which a robot
interacts with an \emph{uncertain} environment and gathers information by visiting locations. 
The goal  is to minimize its expected travel cost to 
cover a given submodular function. Adaptive solutions, 
where the robot  incorporates all available information to select the next location to visit, 
achieve the best objective. 
However, such a solution is resource-intensive as it entails recomputing after every visited location.
A more practical approach is to design solutions with a small number of adaptive ``rounds'', where the robot  recomputes  only once at the start of each round. 
In this paper, 
we design an algorithm for \prob parameterized by the number $k$ of adaptive rounds, and prove a  smooth trade-off between $k$ and the solution quality (relative to fully adaptive solutions).
We validate our theoretical results  by experiments on a real road network, where we observe that a few rounds of adaptivity suffice to obtain solutions
of cost almost as good as fully-adaptive ones.
\end{abstract}
\section{Introduction}
We consider the \emph{informative path planning} (\prob) problem 
in which a robot
interacts with an \emph{uncertain} environment and gathers information by visiting locations. 
The informative path planning problem has been widely studied, and has applications in 
information gathering \cite{SinghKG+06},
object detection \cite{PlattKL+11}, and 
manipulating a robot arm for tasks like 
pushing a button or grasping \cite{JavdaniCK+14}.
We discuss two applications of \prob.

First, consider the following disaster management application. 
Suppose that an autonomous unmanned aerial
vehicle (UAV) 
is searching for a lost victim \cite{LimHL16}.
The UAV acquires new information on the
victim’s location by 
using onboard sensors, and the goal is to 
plan a search strategy in order to find the victim as fast as possible.
Another application of \prob, arising in information gathering, is in 
monitoring algae biomass in a lake \cite{DhariwalZS+06}.
It is not economical to
cover the lake sufficiently with 
\emph{static} sensors, and instead  one wants to plan a route for a
robotic boat (carrying a sensor) to move to various
locations in the lake to gather information.

We note that both these applications 
involve \emph{stochastic submodular optimization}: the uncertainty stems from not knowing the underlying state of the world (the victim's true location, or the concentration of pollutants in the lake), and progress (eliminating possible locations from consideration, or collecting information from different parts of the lake) 
can be captured 
using a submodular function (see \cite{LimHL16} and \cite{SinghKG+06} for details). 
In most works, stochastic submodular optimization is restricted over a set domain; that is,  the goal is to select some subset to optimize the expected objective.
However, we are interested in settings where the robot
interacts with the environment  by visiting different sensing location; in other words,
the 
robot's decisions are constrained to form a \emph{path} (rather than  an arbitrary subset).
So, the goal in \prob is to \emph{minimize its expected travel cost to  
cover a given submodular function} (see \S\ref{sec:defn} for a formal definition).

Solutions
to such a stochastic problem are 
sequential decision processes, making them highly adaptive:
at each step,
the robot incorporates all available information to select the next location to visit.
This process continues until the submodular function is covered.
 However, such a solution entails recomputing after every visited location, which can be  resource-intensive. So, fully adaptive solutions may not be feasible in practical situations. This motivates the design of solutions with a small number of adaptive ``rounds'', where the robot  recomputes  only once at the start of each round.
Such solutions strike a balance between achieving the best objective and keeping resource utilization low. For example, from an  energy efficiency perspective: a robot often operates with limited battery capacity, and the need for frequent re-computations after each visited location can lead to substantially faster energy consumption. Additionally, the process of incorporating observed data and preparing it for re-computation is not instantaneous as it involves data gathering/analysis, and may lead to increased computational time. 
By keeping the  rounds of adaptivity small, we aim to mitigate 
such operational challenges.  
We note that the trade-off between rounds of adaptivity and solution quality is not new, and has been studied in various streams of literature (see \S\ref{sec:related-work}).
We make the following contributions.

\vspace{-3mm}
\begin{enumerate}
    \item We design an algorithm for \prob
    parameterized by the number $k$ of adaptive rounds, and prove a smooth trade-off between $k$ and
    the solution quality (relative to fully adaptive solutions). 

    \item  We consider separately an important special case of \prob: path planning for hypothesis identification \cite{LimHL16}, and obtain a   better performance guarantee via a more   efficient algorithm.
    
    \item Finally, we run computational experiments on 
    a real road network dataset and
    previously-used instances of  hypothesis identification.  
    For these instances, we observe that with $2$ rounds of adaptivity, the cost is on average within $50\%$ and $12\%$ of the fully adaptive cost, respectively. 
Moreover, the $2$-round algorithm is on average $15$ times faster than the fully adaptive one.  
\end{enumerate}

\vspace{-2mm}
\subsection{Definitions}\label{sec:defn}
\vspace{-2mm}

An instance of \prob is given by the tuple $(X, r, d, M, \D, O, f)$, 
where $X$ corresponds to a finite set of 
$n$ sensing locations,
$r$ is the initial location of the robot, and $d$ is a metric on $X \cup \{r\}$. We assume throughout that $d$ is symmetric and satisfies triangle inequality. 
Here $M = \{1, \ldots, m\}$ denotes a  finite set of \emph{hypotheses} or \emph{scenarios}. We take the Bayesian approach, and  
use $p_\scn$ to denote the prior  probability of each hypothesis $\scn\in M$, where $\sum_{\scn\in M} p_\scn =1$. 
The set $O$ denotes all  possible observations:  each location $v  \in X$ realizes to a random observation in 
$O$. 
The distribution $\D$ specifies the probability $p_\scn$ of each hypothesis $\scn\in M$ as well as the observations $\{\scn(v) \in O : v\in X\}$ at all locations {\em under hypothesis} $\scn$. The true hypothesis   $\scn^*$ is drawn from $  M$ according to the distribution $\D$; however $\scn^*$ is initially unknown to the algorithm. When location $v\in X$ is visited, the robot observes $\scn^*(v)\in O$ and can use this information to update its priors.

If  the robot has visited a subset $S\sse X$ of locations and observed $o_v\in O$ at each $v\in S$ then the set $\{(v,o_v) : v\in S\}\sse X\times O$ is called a {\em partial realization}. We use $\psi\sse X\times O$ to denote a generic partial realization; if we want to additionally specify the locations $S$  contained in $\psi$ then we use the notation $\psi(S)$. We define $\psi_{\scn}(X) := \{(v,\scn(v)): v \in X\}$ to denote the partial realization associated with hypothesis $\scn$ at all locations. 
We say that hypothesis $\scn$ is \emph{compatible} with a partial realization $\psi$, denoted $\scn \sim \psi$, if $\psi \sse \psi_{\scn}(X)$.

Let $\Psi = 2^{X\times O}$ be the power-set of all location-realization pairs. Note that $\Psi$ contains every partial realization. Not every subset in $\Psi$ corresponds to an actual  partial realization, but using the full power-set $\Psi$ makes   the definition of our utility function cleaner. In particular, function $f:\Psi \rightarrow \mathbb{Z}_+$ is a monotone submodular set function. 
{Formally, we say that $f$ is \emph{submodular} if whenever $A \sse B \sse \Psi$ and $e \notin B$, we have $f(A \cup \{e\}) - f(A) \geq f(B \cup \{e\}) - f(B)$, and we say that $f$ is \emph{monotone} if whenever $A \sse B$, we have $f(A) \leq f(B)$.}
Let $Q\in \mathbb{Z}_+$ be the maximal value of the function. We say that a partial realization $\psi\in \Psi$ {\em covers} function $f$ if $f(\psi)=Q$.
 We assume that $f(\psi_{\scn} (X)) = Q$ for each hypothesis $\scn \in M$. In other words, if we visit all locations then the function $f$ will be covered, irrespective of the true hypothesis $\scn^*$.

Let  $\Pi = (r, v_1, \ldots, v_{\ell}, r)$ be any tour that starts and ends at $r$; we abuse notation and also use $\Pi$ to denote the set of locations visited.
All tours in this paper will begin/end at $r$: we will not state this condition each time (to avoid clutter).  We say that tour $\Pi$  \emph{covers}  hypothesis $\scn$ if,
and only if, $f(\psi_\omega(\Pi)) = Q$. The goal in \prob is to design a tour $\Pi$ (possibly adaptively) 
that covers the true hypothesis $\scn^*$. An equivalent condition is that the observed partial realization $\psi(\Pi)$ at the end of  the tour must satisfy $f(\psi(\Pi))=Q$. The 
  objective is to minimize the expected 
distance $\E[d(\Pi)]$ of the  tour $\Pi$, where the expectation is taken over $\scn^*$.
We note that an adaptive tour $\Pi$  decides on the next location to visit based on the observations at all previous locations.  We are interested in solutions that have limited adaptivity as defined next.

\begin{definition}\label{def:k-round}
For an integer $k\ge 1$, a {\bf $k$-round} solution
proceeds in $k$ rounds of adaptivity. In each round
$\ell\in \{1,\ldots, k\}$, the solution specifies a tour on all remaining locations and visits them in this order until some stopping rule (at which point  it starts  the next round).
The tour  in round $\ell$ can depend on the observations seen in  rounds $1,\ldots, \ell-1$.
\end{definition}
\vspace{-2mm}
In a $k$-round solution,   tour re-computation only occurs  at the start each round, which happens at most $k$ times. 
Setting $k=1$ in Definition~\ref{def:k-round} gives us a \emph{non-adaptive} tour that does not have to recompute after it starts. On the other hand, setting $k\rightarrow \infty$ (effectively  $k=n$) gives us a \emph{fully  adaptive} tour, that recomputes after each visited location.
Having more rounds potentially leads to a smaller objective value, so fully adaptive solutions have the least objective value. Our performance guarantees are relative to an optimal fully adaptive solution; 
let $\OPT$ denote this solution and its cost. 
The \emph{$k$-round-adaptivity gap} is defined as follows:
$$\sup_{\text{instance } I} \frac{\E[\text{cost of best $k$-round  solution on }I]}{\E[\text{cost of best fully adaptive solution on }I]}.$$ 

In formulating \prob, we required the solution to be a {\em tour} originating from $r$. 
We note that one could also ask for a {\em path}   originating from $r$, which is allowed to end at any location as long as the function $f$ is covered.  All our results also apply to this path variant, formalized in the following proposition. 
\begin{proposition}\label{prop-tour-path}
Any $\alpha$-approximation algorithm for the tour version of \prob gives a $2\alpha$-approximation to the path version of \prob.
\end{proposition}
Indeed, by symmetry and triangle inequality, the cost of an optimal tour is at most $2$ times the cost of an optimal path, and any $\alpha$-approximate solution  for the tour version of \prob is also feasible  for the path version.

\subsection{Results and Techniques}

Our first result is a $k$-round algorithm for \prob. 
\begin{theorem}\label{thm:main-k-round}
For any integer $k \geq 1$ and constant $\epsilon > 0$, there is a $k$-round adaptive algorithm for the 
informative path planning
problem with cost 
at most $\bigOh\left(\log^{2+\epsilon}(n) \cdot m^{1/k}\cdot (\log m + k\log Q)\right)$ times the cost of an optimal
adaptive algorithm.
\end{theorem}
In each round of the algorithm, we 
build a 
\emph{non-adaptive} tour using observations from previous rounds that partially covers function $f$.  
This tour is built as follows: we iteratively compute a ``score'' for each tour and greedily select a tour based on these scores. 
The algorithm then visits sensing locations according to the constructed tour until a stopping condition is met.
A similar idea was used in \cite{GGN21} in obtaining $k$-round adaptive algorithms for the ``scenario submodular cover'' problem, where they computed a score for each item, and greedily (selecting an item of maximum score) built a non-adaptive list of items. 
Since our solutions are constrained to be tours, simply scoring each sensing location does not work: we need an appropriate scoring function for paths/tours.
Doing this naively, however, causes the search space 
to blow-up from $n$ to $n!$ as we need to search for a tour rather than a single location. Crucially, 
 we show  that
selecting a tour to maximize the score function turns out to be
an instance of ratio submodular orienteering (see \S\ref{sec:k-round}, Definition~\ref{def:rso}). This is an NP-hard problem, for which an $\bigOh(\log^{2+\epsilon}n)$-approximation algorithm is known~\cite{CalinescuZelikovsky05}. So,  
at each step we pick a tour that \emph{approximately} maximizes the score function,
which is subsequently
appended to the non-adaptive tour. This ``partial covering'' result is summarized in Theorem~\ref{thm:partial-cover-thm}. Then, we design each round so that it guarantees a suitable  measure of progress (roughly, eliminating $m^{1/k}$ fraction of hypotheses), and the overall $k$-round algorithm fully covers the function $f$.

The high-level analysis for Theorem~\ref{thm:main-k-round} 
is similar to that in \cite{GGN21} for  scenario submodular cover. However, there are some key differences in handling paths/tours rather than individual items. 
As with a number of stochastic covering problems, e.g., \cite{INZ12,CuiNagarajan23},  we relate the “non-completion” probabilities of the algorithm after cost 
$\alpha \cdot t$ 
to the optimal
adaptive solution after cost $t$ (for all $t\ge 0$).
The factor $\alpha$ (specified later)  appears in the final approximation guarantee (along with the approximation for ratio submodular orienteering).
The $m^{1/k}$ dependence in the $k$-round adaptivity gap from Theorem~\ref{thm:main-k-round} is the best possible even in the special case of a star metric: this follows from the lower bound result for  scenario submodular cover in \cite{GGN21}.

Next, we consider an important special case of \prob: path planning for hypothesis identification (\probh). Instead of covering an arbitrary submodular function, 
the goal in \probh is to identify the underlying hypothesis $\scn^*$ by visiting locations at the minimum expected distance. 
We obtain the following improved result for \probh.
\begin{theorem}\label{thm:k-round-id}
For any integer $k \geq 1$, there is a $k$-round adaptive algorithm for the 
hypothesis identification
problem with cost 
at most $\bigOh\left(\log^{2}(n) \cdot m^{1/k}\cdot k\cdot \log m)\right)$ times the cost of an optimal
adaptive algorithm.
\end{theorem}
The main idea here is to exploit the special structure of the submodular function $f$ corresponding to  hypothesis identification. We show that for this special function $f$, maximizing the score-function corresponds to solving the {\em ratio group Steiner} problem (rather than  ratio submodular orienteering).  This allows us to  use a better $\bigOh(\log^2n)$ approximation algorithm for ratio group Steiner \cite{CharikarCGG98}. We actually present a self-contained algorithm for ratio group Steiner  that is simpler (and easier to implement) than the one in \cite{CharikarCGG98}.

In fact, using our  partial covering algorithm (Theorem~\ref{thm:partial-cover-thm}) and a  different  measure of progress in each round  (as in Theorem 6.7 of \cite{GGN21}), we can also obtain $2k$-round algorithms with better approximation guarantees of $\bigOh\left(\log^{2+\epsilon}(n) \cdot m^{1/k} \cdot \log(Q m)\right)$ for \prob and $\bigOh\left(\log^{2}(n) \cdot m^{1/k}\cdot \log m)\right)$ 
for \probh \ {(see \S\ref{app:better-appx} for details)}. Setting the number of rounds to $\bigOh(\log m)$, we then 
 get  approximation guarantees of $\bigOh\left(\log^{2+\epsilon}(n) \cdot \log(Q m)\right)$ and $\bigOh\left(\log^2(n) \cdot \log m\right)$ for \prob and \probh respectively. These approximation ratios match the previous-best approximation ratios for these problems, even for fully-adaptive algorithms~\cite{NavidiKN20,GNR17}. In fact, \probh and \prob generalize the {\em group Steiner tree} problem~\cite{GargKR00}, for which the best known approximation ratio is $\bigOh\left(\log^2(n) \cdot \log m\right)$; there is also an $\Omega(\log^{2-\epsilon} n)$ hardness of approximation \cite{HalperinK03}.

\subsection{Related Work}\label{sec:related-work}

It is known that metrics of informativeness in several domains (for example,  
sensor placement \cite{KrauseGuestrin05} and target search \cite{HollingerSD+09}) exhibit \emph{submodularity}. 
Submodular set function optimization has been studied extensively \cite{W82, NemhauserWF78}, and has also been extended to 
optimizing over paths \cite{ChekuriPal05, CalinescuZelikovsky05}.
Consequently, any approximation algorithm for \emph{submodular path orienteering} \cite{ChekuriPal05} can be used to 
plan a path for a robot in order to
maximize a submodular function of the visited locations.
\cite{SinghKG+06} provided an approach
for extending any single robot
algorithm
to the multi-robot
setting, with a (nearly) matching approximation guarantee.

Submodular  optimization (over sets) has been extended to the stochastic setting in a number of works, e.g.,~\cite{AsadpourN16,GolovinK-arxiv,INZ12,GHKL16,GuptaNS17}. Recent works \cite{AAK19}, \cite{EsfandiariKM19}, \cite{GGN21} are particularly relevant to us: these papers establish trade-offs between  {\em rounds of adaptivity} and the approximation factor for  stochastic submodular cover problems, where one wants to select a {\em subset} to cover a submodular function  (there are different settings with independent, scenario-based and adaptive-submodularity conditions). Our work extends the results from \cite{GGN21} for scenario-based distributions  to the case of optimizing over  {\em paths} in a metric.

Stochastic submodular optimization over paths has also received significant attention. 
\prob has been studied in robotics and related fields, and
many heuristic approaches have been proposed to solve the problem.
For example, in \cite{HollingerEH+13},  a minimum-cost tour is constructed on ``informative'' sensing locations,  
and in \cite{HollingerMS11}, the idea is to search for 
a strategy over a finite planning horizon.
An adaptive approach appeared in \cite{SinghKK09}: their algorithm re-plans every step using
a non-adaptive information path planning algorithm.

The special case of \probh has itself been studied widely.
This appears in \cite{GNR17} as the ``isolation problem'' enroute to obtaining approximation algorithms for the \emph{adaptive traveling salesman} problem.
\cite{GNR17} obtained a fully-adaptive $\bigOh(\log^2n\,\log m)$-approximation algorithm for \probh. 
\cite{LimHL15} and \cite{LimHL16} obtained similar algorithms for \prob; these hold for a slightly more general  definition involving adaptive-submodularity \cite{GolovinK-arxiv}. When applied to \prob, the algorithms in \cite{LimHL15} and \cite{LimHL16} yield a fully adaptive $\bigOh\left(\log^{2+\epsilon}(n)\cdot \log(m) \cdot \log\left({1}/{p_{\min}}\right)\right)$-approximation algorithm; here $p_{min}\le 1/m$ is the minimum probability of any hypothesis. \cite{NavidiKN20} obtained an improved 
$\bigOh\left(\log^{2+\epsilon}(n)\cdot \log(m) \right)$-approximate fully adaptive algorithm for \prob.

The trade-off between rounds of adaptivity and solution quality has also been considered in other contexts. \cite{GaoHR+19,EsfandiariKM+21}, and \cite{AgarwalGN22} study   online learning problems, where observations are made in batches.  \cite{BalkanskiBS18,BalkanskiS18b,BalkanskiRS19}, and \cite{ChekuriQ19} study
deterministic submodular optimization, where function queries are batched. However, the techniques used in these papers are completely different from ours.

\section{$k$-Round Algorithm for \prob}\label{sec:k-round}
\vspace{-2mm}

In this section, 
we design a $k$-round algorithm for the 
{\it informative path planning} problem (\prob) and prove Theorem~\ref{thm:main-k-round}. A key component of our algorithm is a non-adaptive algorithm to solve a \emph{partial cover} version of \prob. 
Formally, an instance of the partial cover version of \prob is the same as an instance of \prob with an additional parameter $\delta \in (0, 1]$. 
Now, the goal is to visit a set of locations $T$ that realize to $\psi(T)\in \Psi$ such that either 
(i) number of compatible scenarios  $|\{\scn\in M: \psi(T) \sse \psi_{\scn}\}|<\delta m$, or (ii) the function $f$ is fully covered, i.e., $f(\psi(T)) = Q$. 
The $k$-round algorithm for \prob will then recursively solve the partial cover version with carefully chosen values for the parameter $\delta$.

An important subroutine in our algorithm is the following {\em deterministic} problem.
\begin{definition}[Ratio Submodular Orienteering (\RSO)]\label{def:rso}
    Given a metric $d$ on locations $X\cup\{r\}$ and a monotone submodular function $g:2^X \rightarrow \mathbb{Z}_+$, find an $r$-tour $T$ 
    that maximizes the ratio $\frac{g(T)}{d(T)}$, where $g(T)$ is the function value on the nodes of $T$ and $d(T)$ is the total distance in $T$.
\end{definition}
This problem is NP-hard, but there are poly-logarithmic approximation ratios known. In particular, \cite{CalinescuZelikovsky05} gave a
$\bigOh(\log^{2+\epsilon}n)$-approximation algorithm with runtime $n^{\bigOh(1/\epsilon)}$, where $\epsilon>0$ is a constant.
If one allows for quasi-polynomial time $n^{\bigOh(\log n)}$ then a better $\bigOh(\log n)$-approximation algorithm is known \cite{ChekuriPal05}.   
It is also hard to approximate to a factor better than $\bigOh(\log^{1-\epsilon} n)$ \cite{HalperinK03}.

\begin{theorem} \label{thm:partial-cover-thm}
There is a non-adaptive algorithm for the partial cover version of \prob  
with expected cost 
$\mathcal{O} \left(\frac{\rho}{\delta}\log\left(\frac{Q}{\delta}\right)\right)$ times the cost of the optimal adaptive solution for \prob,
where $\rho$ is the best  approximation guarantee for ratio submodular orienteering.
\end{theorem}

The algorithm creates a
\emph{pre-planned} (non-adaptive) tour;
that is, without
knowing the realizations at the locations. 
We find that iteratively selecting tours that maximize a carefully-defined score function (see Equation~\eqref{eqn:greedy-choice}) works well;
however, selecting such tours turns out to be an NP-hard problem. 
So,
at each step we pick a tour that \emph{approximately} maximizes the score function,
which is subsequently
appended to the non-adaptive tour: this process continues until all 
locations are included in the non-adaptive tour, or we can conclude that the number of compatible scenarios after visiting the already selected locations will be less than $\delta m$ (see Definition~\ref{defn:partn-scn}).
We note that the score of a
tour (roughly) measures
the progress we can makes
towards (i) eliminating scenarios and (ii) covering function $f$ on visiting the tour. Crucially, we prove that the numerator of this score function corresponds to a  monotone submodular function (see Lemma~\ref{lem:score-submod}). So, we can use an approximation algorithm for  \RSO to optimize the score. 
Before we state the score
function, we need some definitions.

\begin{definition} \label{defn:partn-scn}
For any  ${S}\sse X$, let $\H({S})$
denote the partition $\{ Y_1, \cdots, Y_{\ell}\}$ of the scenarios $M$ where
all scenarios in a part have the same realization for the locations in ${S}$. 
Let $\mathcal{Z} := \{ Y \in \H({S}) : |Y| \geq \delta |M|\}$ 
be the set of ``large'' parts.
\end{definition}

Consider scenarios $\scn_1$ and $\scn_2$. According to Definition~\ref{defn:partn-scn}, $\scn_1$ and $\scn_2$ belong to the same part of $\H(S)$ if and only if $\scn_1(v)=\scn_2(v)$ for all $v \in S$; i.e., visiting the locations in $S$ leads to the {\em same} partial realization under  either scenario $\scn_1$ or $\scn_2$.
After observing the realization of $S$, the set of compatible scenarios must be one of the parts in $\H(S)$. 
Also note that $|\mathcal{Z}| \leq \frac{1}{\delta}$ as each part in ${\cal Z}$ has at least $\delta|M|$ scenarios.

\begin{definition}\label{defn:parts-item}
For any location $v \in X$ and subset $Z \sse M$ of scenarios, consider the partition 
of $Z$ based on the realization of $v$. 
Let $B_v(Z) \sse Z$ be the largest cardinality 
part, and define $L_v(Z) := Z \setminus B_v(Z)$.
\end{definition}

The above definition is used to quantify the ``information gain'' of visiting a single location. 
If the realized scenario $\scn^* \in L_v(Z)$, then
we can eliminate at least half the scenarios in $Z$ by visiting location $v$. For any part $Z \in {\cal H}(S)$, note that the partial realizations $\psi_\scn(S)$
are identical for all $\scn \in Z$: we use 
$\psi_Z(S) \sse X \times O$ to denote this partial realization. 

Let $\Pi$  denote the non-adaptive tour constructed so far in our algorithm, and let   $S$ be the set of locations in $\Pi$. 
The score \eqref{eqn:greedy-choice} of a new tour $T$ is computed by considering two notions of progress for each $Z \in \mathcal{Z}$:

\begin{itemize}
    \item \emph{Information gain} $\sum_{\scn \in L_T(Z)} p_{\scn}$,    
    measures the total 
    probability of the scenarios that belong to $L_v(Z)$ for some $v \in T$.
 
    \item \emph{Relative function gain} $\sum_{\scn \in Z} p_{\scn} \cdot \frac{f(\psi_Z(S) \cup \psi_{\scn}(T)) - f(\psi_Z(S))}{Q - f(\psi_Z(S))}$   measures the expected relative gain obtained by
    visiting locations in tour $T$ (expectation is w.r.t. scenarios in $Z$).
\end{itemize}
\vspace{-2mm}
The overall score of tour $T$ is the sum of these terms (over all parts in $\mathcal{Z}$) normalized by the distance $d(T)$ of the tour.
Note that 
the score of a tour
is computed only using 
“large” parts $\mathcal{Z}$. 
This is because if the realization of $S$
corresponds to any other part then the number of compatible scenarios would be less than $\delta m$ (and
the partial cover algorithm would have terminated). We show that the numerator of \eqref{eqn:greedy-choice} is a monotone and submodular function.

\begin{lemma}\label{lem:score-submod}
    Let 
    \begin{align*}
    g(T) &=  \sum_{Z \in {\cal Z}}\Bigg( \sum_{\scn \in L_T(Z)} p_{\scn} + \sum_{\scn \in Z} p_{\scn} \cdot \frac{f(\psi_Z(S) \cup \psi_{\scn}(T)) - f(\psi_Z(S))}{Q - f(\psi_Z(S))} \Bigg), 
    \end{align*}
    where $S \sse X$ is some fixed subset. Then $g$ is monotone and submodular.
\end{lemma}

\begin{proof}[Proof of Lemma~\ref{lem:score-submod}]
    Fix $Z \in \mathcal{Z}$, and let $$ g_1(T, Z) = \sum_{\scn \in L_T(Z) } p_{\scn} \qquad \qquad \text{ and }  \qquad \qquad  g_2(T, Z) = \sum_{\scn \in Z} p_{\scn} \cdot \frac{f(\psi_Z(S) \cup \psi_{\scn}(T)) - f(\psi_Z(S))}{Q - f(\psi_Z(S))}. $$
    Since addition preserves monotonicity and submodularity, 
    it suffices to show that 
    $g_1(\cdot, Z)$ and $g_2(\cdot, Z)$ are both monotone and submodular for all $Z$. Note that the monotonicity and submodularity of $g_2(\cdot, Z)$ follows from the monotonicity and submodularity of $f$ (since $S$ is fixed).
    Towards proving the monotonicity and submodularity of $g_1(\cdot, Z)$,
    recall that $L_T(Z) = \cup_{v \in T} L_v(Z)$. Suppose $T \sse \overline{T}$. Then, by definition, we have $L_T(Z) \sse L_{\overline{T}}(Z)$ which implies $g_1(T, Z) \leq g_1(\overline{T}(Z))$. Furthermore, consider $v \notin \overline{T}$. Then, $ L_v(Z) \setminus L_{\overline{T}}(Z)  \sse  L_v(Z) \setminus  L_{{T}}(Z)$ which implies $g_1(\overline{T} \cup \{v\}, Z) - g_1(\overline{T}, Z) \leq g_1({T} \cup \{v\}, Z) - g_1({T}, Z)$. Thus, we can conclude that $g_1(\cdot, Z)$ is monotone and submodular, which concludes the proof.  
\end{proof}

\begin{algorithm}[t]
\caption{Partial Covering Algorithm $\mathtt{PCA}((X, r, d, M, \D, O, f), \delta)$} \label{alg:partial-cover}
\begin{algorithmic}[1]
  \State $S \leftarrow \emptyset$, $\Pi \leftarrow \emptyset$
\While{$S \neq X$}
\State Define $\H({S})$, ${\cal Z}$ and $L_v(Z)$ as in Definitions~\ref{defn:partn-scn} and \ref{defn:parts-item}
\If{$\mathcal{Z}$ is empty} \ \textbf{break}
\EndIf
\State Select  tour ${T}$ that $\rho$-approximately maximizes:

\begin{align}\label{eqn:greedy-choice}
    &\text{score}(T) = \frac{1}{d(T)} \cdot \sum_{Z \in {\cal Z}}\Bigg( \sum_{\scn \in L_T(Z)} p_{\scn} +\\& \notag \sum_{\scn \in Z} p_{\scn} \cdot \frac{f(\psi_Z(S) \cup \psi_{\scn}(T)) - f(\psi_Z(S))}{Q - f(\psi_Z(S))} \Bigg)
\end{align}

where $L_T(Z) = \cup_{v \in T}L_v(Z)$.
\State $S \gets S \cup T$, $\Pi \gets \Pi \circ T$
\EndWhile
\State  $R \gets \emptyset$, $\psi(R) \gets \emptyset$, $H\gets M$.
\While{$|H|\geq \delta m$ and $f(\psi) < Q$} 
\State $T \gets$ first tour in $\Pi$ not yet visited
\State $\psi(T) \gets$ realization of vertices in $T$
\State $R \gets R \cup T$, $\psi(R) \gets \psi(R) \cup \psi(T)$
\State $H \gets \{ \scn \in H: \scn \sim \psi(R)\}$; that is, set of compatible scenarios
\EndWhile
\State return visited locations $R$, partial realization $\psi(R)$ and compatible scenarios $H$.
\end{algorithmic}
\end{algorithm}

Note that $\text{score}(T) = g(T)/d(T)$, and since $g(T)$ is monotone and submodular, we can use 
an approximation algorithm for \RSO to optimize the score.
Once the non-adaptive tour $\Pi$, which itself is a concatenation of many smaller tours, is specified, the algorithm starts by visiting tours in this order until
(i) the number of compatible scenarios drops below $\delta m$, or (ii) the realized
function value equals $Q$. Note that in case (ii), the function is fully covered. See Algorithm~\ref{alg:partial-cover} for a
formal description of the non-adaptive algorithm.

We recursively use this 
non-adaptive partial cover 
algorithm to get a  $k$-round solution for   \prob. The first round involves setting  $\delta = m^{-1/k}$ in $\mathtt{PCA}$. 
At the end of  round $\#1$,  let $R$ be the set of locations visited, $\psi = \psi(R)$ be   the partial realization observed, and   $H\sse M$ be the  compatible scenarios.
Then, we can condition on the scenarios in $H$, and define a ``residual'' function $f_{\psi}:\Psi\rightarrow \ZZ_+$ as $ f_{\psi}(\phi) = f(\psi \cup \phi) - f(\psi)$, which is also monotone and submodular. Finally, we recurse on this residual function $f_{\psi}$ to get a $k-1$ round solution.  See Algorithm~\ref{alg:k-rounds} for a formal description. We formalize this discussion in the following result.

\begin{algorithm}[t]
\caption{$k$-round adaptive algorithm for \prob, $k\text{-}\mathtt{ADAP}(\mathcal{I} = (X, r, d, M, \D, O, f), k)$} \label{alg:k-rounds}
\begin{algorithmic}[1]
\State Run $\mathtt{PCA}(\mathcal{I}, m^{-1/k})$ for the first round. Let $R$ denote the set of locations visited. Let $\psi$ and $H$ denote the partial realization and set of compatible scenarios respectively.

\State Define the residual submodular function $f_{\psi}(\phi) = f(\psi \cup \phi) - f(\psi)$, and define the distribution $D_H$ by conditioning on the remaining scenarios $H$.

\State Solve $k\text{-}\mathtt{ADAP}(\widehat{\mathcal{I}} = (X \setminus R, r, d, H, \D_H, O, f_{\psi}), k-1)$
\end{algorithmic}
\end{algorithm}

\begin{theorem}\label{thm:k-round-ipp}
\sloppy Algorithm~\ref{alg:k-rounds} is a $k$-round algorithm for 
\prob  with expected cost 
$\bigOh\left(\rho \cdot m^{1/k}\cdot (\log m + k\log Q)\right)$ times the   optimal fully adaptive cost. Here, $m$ is the number of scenarios and  $\rho$ is the  approximation guarantee for \RSO. 
\end{theorem}

\begin{proof}[Proof of Theorem~\ref{thm:k-round-ipp}.]
We prove the theorem by induction on $k$. 
Let $\mathcal{I} = (X, r, d, M, \D, O, f)$ denote the instance of \prob, and let $\OPT$ denote the expected cost of an optimal fully adaptive solution for $\mathcal{I}$. 
In the base case where $k = 1$, we set $\delta = 1/m$. 
Let $R$ be the set of locations visited by $\mathtt{PCA}(\mathcal{I}, m^{-1})$; let $\psi(R)$ denote the corresponding partial realization and $H$ the set of compatible scenarios. By Theorem~\ref{thm:partial-cover-thm}, either (i) $|H| < m^{-1} \cdot m = 1$, or (ii) $f(\psi(R)) = Q$. Since case $(i)$ cannot happen, we must have $f(\psi(R)) = Q$; that is, the realized scenario is fully covered. Furthermore, the expected cost of the solution (again, by Theorem~\ref{thm:partial-cover-thm}) is at most $\bigOh(\rho \cdot m \cdot (\log m + \log Q)) \cdot \OPT$, which proves the base case.

Now assume that $k > 1$. Recall that we set $\delta = m^{-1/k}$, and by executing $\mathtt{PCA}(\mathcal{I}, m^{-1/k})$, we get a tour visiting locations $R$ with partial realization $\psi(R)$ and remaining compatible scenarios $H$ such that (i) $|H| < \delta m = m^{(k-1)/k}$, or (ii) $f(\psi(R)) = Q$. 
By Theorem~\ref{thm:partial-cover-thm}, the total cost incurred in the first round is $\bigOh(\rho \cdot m^{1/k} \cdot (\frac{1}{k} \log m + \log Q)) \cdot \OPT$. If we are in case (ii) above; that is, if $f(\psi(R)) = Q$, then we are done and the algorithm incurs no further cost. 
So, we can assume that case (i) holds instead. 
Let $\widehat{{\cal I}}$ denote the residual instance,  $\widehat{m} = |H|<m^{1-1/k}$ the number of scenarios, $\widehat{Q} = Q - f(\psi(R))$ the target,
and  $\widehat{\OPT}$ the optimal adaptive cost of $\widehat{{\cal I}}$; note that these parameters are conditional on the observations $\psi(R)$ in round $\#1$. 
By the inductive hypothesis, Algorithm~\ref{alg:k-rounds} gives a  $(k-1)$-round solution with expected cost $\bigOh(\rho \cdot \widehat{m}^{1/(k-1)} \cdot (\log \widehat{m} + (k-1) \log \widehat{Q})) \cdot \widehat{\OPT}$. Note that the optimal adaptive solution for the original instance $\cal{I}$ provides a feasible (adaptive) solution for the residual instance by just restricting to the scenarios in $H$. Moreover, taking an expectation over $\widehat{{\cal I}}$, we get $\E[\widehat{\OPT}] \leq \OPT$. 
So, the expected cost in the remaining $(k-1)$-rounds is upper bounded by
\begin{equation*}\label{eq:k-round-thm-conditonal-cost} \bigOh(\rho \cdot \widehat{m}^{1/(k-1)} \cdot (\log \widehat{m} + (k-1) \log \widehat{Q})) \cdot \E[\widehat{\OPT}] = \bigOh(\rho \cdot {m}^{1/k} \cdot (\log {m}^{\frac{k-1}{k}} + (k-1) \log Q)) \cdot \OPT,\end{equation*}
where we used $\widehat{m}  <m^{1-1/k}$. 
Finally, combining this with the cost of the first round, we get that the total cost incurred by our algorithm is
$$\bigOh\left(\rho \cdot m^{1/k} \cdot \left(\frac{1}{k} \log m + \log Q\right)\right) \cdot \OPT + \bigOh\left(\rho \cdot {m}^{1/k} \cdot \left(\frac{k-1}{k}\log {m}  + (k-1) \log Q\right)  \right) \cdot \OPT,$$ 
which equals $\bigOh\left(\rho \cdot m^{1/k}\cdot (\log m + k\log Q)\right)$ 
as desired.
\end{proof}

Combined with the $\bigOh(\log^{2+\epsilon}n)$ approximation algorithm for submodular orienteering \cite{CalinescuZelikovsky05}, we have Theorem~\ref{thm:main-k-round}.

\subsection{Proving Theorem~\ref{thm:partial-cover-thm}}
For the analysis, we 
denote our non-adaptive policy, and its (random) cost as $\NA$. Similarly, we use $\OPT$ to refer to an optimal fully adaptive policy and its (random) cost.
We
refer to the cumulative cost incurred by either policy as \emph{elapsed time}.

We define constants $\beta$ (specified later)  and $L \coloneqq \log \left(\frac{Q}{\delta}\right)$. 
Next, we define terms that are used to track the 
progress of $\OPT$ and $\NA$ respectively.
 
\begin{itemize}
    \item $o(t) \coloneqq \pr \left(\OPT \text{ does not terminate by time } t \right) $
    \item $a(t) \coloneqq\pr \left(\NA \text{ does not terminate by time } \beta L t\right)$
\end{itemize}

Observe that $o(t)$ and $a(t)$ are non-increasing functions of $t$, and 
$o(0) = a(0) = 1$. 
We can view $o(t)$ and $a(t)$ as ``non-completion'' probabilities of $\OPT$ and $\NA$ respectively.
The following lemma 
relates these non-completion probabilities, and forms the crux of our analysis.
\begin{lemma} \label{lem:key-lemma}
For any $i \geq 0$, we have 
\begin{equation}\label{eq:key-eq}
    \frac{\beta\delta}{4\rho}  \cdot \mathlarger{\sum}_{j \geq i} \left(\frac{a(j+1) - 2\cdot o(j+1)}{j+1}\right) \leq a(i).
\end{equation}
\end{lemma}

Using this lemma, we can immediately prove Theorem~\ref{thm:partial-cover-thm}.

\begin{proof}[Proof of Theorem~\ref{thm:partial-cover-thm}]
Using the integral identity for expectations, we can write the expected cost of our non-adaptive policy as follows.
\begin{equation*}
    \E[\NA] = \int_{0}^{\infty} \pr(\NA > t) dt = \int_{0}^{\infty} a\left(\frac{t}{\beta L}\right) dt = \beta L \int_{0}^{\infty} a(t) dt
\end{equation*}
where the final equality follows by applying a change of variables. 
Since $a(t)$ is non-increasing in $t$, we have 
\begin{equation}\label{eq:na-exp}
    \E[\NA] = \beta L \int_{0}^{\infty} a(t) dt \leq \beta L \sum_{i \geq 0} a(i) = \beta L \cdot A
\end{equation}
where we set $A = \sum_{i \geq 0} a(i)$. Similarly, we let $O = \sum_{i \geq 0} o(i)$, and 
sum \eqref{eq:key-eq} over $i \geq 0$ to obtain
\begin{equation*}
        \frac{4 \rho}{\beta \delta}\cdot A = \frac{4 \rho}{\beta \delta} \cdot \mathlarger{\sum}_{i \geq 0} a(i) \geq \mathlarger{\sum}_{i \geq 0}\mathlarger{\sum}_{j \geq i} \left(\frac{a(j+1) - 2\cdot o(j+1)}{j+1}\right) = \sum_{j \geq 1} (a(j) - 2\cdot o(j)) \geq A - 2O
\end{equation*}
where the final inequality uses $a(0) = o(0) = 1$. On rearranging the above inequality, we obtain
\begin{equation}\label{eq:relate-na-opt}
A \leq \frac{2\beta\delta} {\beta\delta - 4\rho} \cdot O.
\end{equation}
Finally, we can write the expected cost of $\OPT$ in terms of $O$ as follows.
\begin{equation}\label{eq:opt-exp}
    O - 1 = \sum_{j \geq 1}o(j) \leq \int_{0}^{\infty} o(t) dt = \E[\OPT].
\end{equation}
where the inequality holds since $o(\cdot)$ is non-increasing.

On combining Equations~\eqref{eq:na-exp},\eqref{eq:relate-na-opt}, and \eqref{eq:opt-exp} we get 
\begin{equation*}
    \E[\NA] \leq \beta L \cdot A \leq \frac{2\beta^2 L \delta}{\beta \delta - 4\rho} \cdot O \leq \frac{2\beta^2 L \delta}{\beta \delta - 4\rho} \cdot \left(\E[\OPT] + 1\right).
\end{equation*}
Setting $\beta = \frac{8\rho}{\delta}$ implies that $\E[\NA] \leq \frac{32L\rho}{\delta}\cdot \left(\E[\OPT]+1\right)$.
We note that the +1 term can be eliminated by a straightforward scaling argument: note that, for any $b \geq 1$, if all costs are scaled by $b$ (resulting in $\OPT$ and $\NA$ being scaled by $b$), we get $\E[\NA] \leq \frac{32L\rho}{\delta}\cdot \left(\E[\OPT]+\frac1b\right)$, thus implying that for large enough $b$, we get $\E[\NA] \leq \frac{32L\rho}{\delta}\cdot \E[\OPT]$, as desired.
\end{proof}

\paragraph{Proof of the Lemma~\ref{lem:key-lemma}.} 
Recall that $\Pi$ is the tour returned by our non-adaptive algorithm, and that $\Pi$ is a concatenation of multiple tours.
For each time $t \geq 0$, let $\Pi_t$ denote the tour being visited at time $t$; that is,
$\Pi_t$ is the tour that causes the cumulative cost to exceed $t$. We say that our non-adaptive policy ($\NA$) is in phase $i$ in the time interval $[\beta L i, \beta L (i+1))$ for any $i \geq 0$.
We define the \emph{total gain} of phase $i$ as $$ G_i := \sum_{t = \beta L i}^{\beta L (i+1)} \text{score}(\Pi_t).$$

\paragraph{Lower bounding $G_i$.}
For lower bounding $G_i$, it is convenient to view the optimal adaptive policy ($\OPT$) as
a single tour, and its cost to be the 
distance until the location
where $f$ is covered for the underlying scenario (doing this only lowers the cost of $\OPT$).

Now, fix some time $t \in [\beta L i, \ \beta L (i+1))$, and let
$S$ denote the set of locations visited prior to the selection of tour $\Pi_t$. 
This does not include the locations that may have been visited on $\Pi_t$.
Define $\H({S})$, and ${\cal Z}$ as in Definition~\ref{defn:partn-scn}.
For each part $Z \in \mathcal Z$, let
$\psi_Z(S)$ denote the realization of $S$ under scenarios in $Z$.
Let $Q_{Z} = Q - f(\psi_Z(S))$ 
denote the 
residual target after visiting locations in $S$,
if the realized scenario is in $Z$, and let $f_Z = f_{\psi_Z(S)}$ be the corresponding residual submodular function. Let $L_v(Z) \sse Z$ be the set of scenarios as in Definition~\ref{defn:parts-item}. Furthermore, define $\OPT_{Z}$ as the sub-tree of $\OPT$ until time $i+1$ when restricted to paths traced by scenarios in $Z$. 
Finally, let 
$\Stem_{Z}$ be the path in $\OPT_Z$ that, at each node $v$, follows the branch corresponding to the realization of the scenarios $B_v(Z) = Z \setminus L_v(Z)$. 
Let $\psi(\Stem_Z)$ denote this partial realization.
We also use $\Stem_Z$ to denote the set of locations on this path. 
Note that each part $Z\in \mathcal{Z}$ is a set of scenarios: we use the definition of $\Stem_Z$ to create a partition of $Z$ as follows.

\begin{itemize}
    \item $Z_{\text{good}} = \{\scn \in Z: \scn \sim \psi(\Stem_Z), \ f \text{ covered, i.e., }f(\psi(\Stem_Z))=Q \}$.
    \item $Z_{\text{bad}} = \{\scn \in Z: \scn \sim \psi(\Stem_Z), \ f \text{ uncovered} \}$.
    \item  $Z_{\text{okay}} = \{\scn \in Z: \scn \in L_v(Z) \text{ for some } v \in \Stem_Z\}$.
\end{itemize}
Recall that $\scn \sim \psi(\Stem_Z)$ means that the scenario $\scn$ is compatible with the partial realization $\psi(\Stem_Z)$; that is, $\psi(\Stem_Z) \sse \psi_{\scn}(X)$.
Using this partition, we classify every part $Z$ as \emph{good}, \emph{okay} or \emph{bad}.
\begin{definition}
We say that part $Z \in \mathcal{Z}$ is {\bf good} if  \(\pr(Z_\text{good}) \geq \pr(Z) / 2\), {\bf okay} if \(\pr(Z_\text{okay}) \geq \pr(Z) / 2\), and {\bf bad} if \(\pr(Z_\text{bad}) \geq \pr(Z) / 2\). 
\end{definition}

\begin{lemma}
Each part $Z \in \mathcal{Z}$ is either good, okay, or bad.
\end{lemma}
\begin{proof}
Observe that either
$\scn \sim \psi(\Stem_Z)$ or
$\scn\in L_v(Z)$ for some $v \in \Stem_Z$. 
Thus, a scenario $\scn$ must either be in $Z_{\text{okay}}$ or in $\left(Z_{\text{good}} \cup Z_{\text{bad}}\right)$. 
Furthermore, suppose that $\scn \sim \psi(\Stem_Z)$: so, it may belong to either 
\(Z_{\text{good}}\) or \(Z_{\text{bad}}\). 
If $f$ is covered, then $\scn \in Z_{\text{good}}$, else $\scn \in Z_{\text{bad}}$. 
This implies that one of \(Z_{\text{good}}\) or \(Z_{\text{bad}}\) must be empty.
Thus, $\max\{\pr(Z_{\text{good}}), \pr(Z_{\text{okay}}), \pr(Z_{\text{bad}})\} \geq \pr(Z)/2$.
\end{proof}

\paragraph{Candidate tour based on $\Stem_Z$.} Note that $\Stem_Z$ is a path originating from $r$. Let $T_Z$ denote the 
tour obtained by returning to     $r$ at the end of $\Stem_Z$. As the distance on $\Stem_Z$ is at most $i+1$, using symmetry and triangle inequality, it follows that the distance $d(T_Z)\le 2(i+1)$. We note that this tour $T_Z$ is only used in our proof (we don't find it in our algorithm). 
    
\begin{lemma}\label{lem:stem_func_gain}
Fix part $Z \in \mathcal{Z}$. If $Z$ is good, then $\sum_{\scn \in Z} p_\scn \frac{f_Z(\psi_{\scn}(T_Z))}{Q_Z} \geq \frac{\pr(Z)}{2}.$
\end{lemma}
\begin{proof}
Let $\scn \in Z_{\text{good}}$. Then, $f_Z(\psi_{\scn}(T_Z)) =f(\psi_{\scn}(\Stem_Z)) = Q$. 
Taking expectations over all scenarios $\scn \in Z$ gives
$$\sum_{\scn\in Z} p_{\scn} \cdot \frac{f_Z(\psi_{\scn}(T_Z))}{Q_Z} \geq \sum_{\scn\in Z_{\text{good}}} p_{\scn} \cdot \frac{f_Z(\psi_{\scn}(T_Z))}{Q_Z} = \sum_{\scn \in Z_{\text{good}}}p_{\scn} \cdot \frac{Q_Z}{Q_Z} \geq \frac{\pr(Z)}{2}$$
where the final inequality follows from the definition of $Z_{\text{good}}$.
\end{proof}

\begin{lemma}\label{lem:info_gain}
Fix part $Z \in \mathcal{Z}$. If $Z$ is okay, then $\pr\left(L_{T_Z}(Z)\right) \geq \frac{\pr(Z)}{2}$ where $L_{T_Z}(Z) = \bigcup_{v \in T_Z}L_v(Z)$.
\end{lemma}
\begin{proof}
It follows from the definition of $Z_{\text{okay}}$ that $Z_{\text{okay}} = \cup_{v \in \Stem_Z}L_v(Z)$. So, $\pr\left(L_{T_Z}(Z)\right) = \pr\left(\cup_{v\in T_Z} L_v(Z)\right) =   \pr(Z_{\text{okay}}) \geq \pr(Z)/2$, where the inequality uses $Z$ is okay.
\end{proof}

Lemmas~\ref{lem:stem_func_gain} and \ref{lem:info_gain} allow 
us to relate the probability of 
parts labeled {good} or {okay}
to the score of the selected tour. 
The following lemma bounds the total 
probability of parts that are either good or okay.

\begin{lemma}\label{lem:noncompletion_prob_bound}
We have $\sum_{Z:\text{okay or } \text{good}} \pr(Z) \geq a(i+1) - 2\cdot o(i+1).$
\end{lemma}
\begin{proof}
Fix $Z \in \mathcal{Z}$, and 
consider $\scn \in Z_{\text{bad}}$.
By definition, $\OPT$ does not cover $f$ when 
$\scn$ is the underlying scenario.
Hence, $\OPT$ costs at least $(i+1)$.
Summing over all parts $Z \in \mathcal{Z}$ (which are disjoint) and scenarios $\scn \in Z_{\text{bad}}$, we get
\begin{equation*}
  o(i+1) \geq \sum_{Z\in \mathcal{Z}}\pr(Z_{\text{bad}}) \geq \sum_{Z: \text{bad}}\pr(Z_{\text{bad}}) \geq \sum_{Z: \text{bad}} \frac{\pr(Z)}{2}
\end{equation*}
where the final inequality follows from $\pr(Z_{\text{bad}}) \geq \pr(Z)/2$ when $Z$ is bad. On rearranging, this yields
\begin{equation}\label{eq:bad-parts}
\sum_{Z: \text{bad}} \pr(Z) \leq 2\cdot o(i+1).
\end{equation}

Next, consider part $Y\in \mathcal H(S) \setminus {\cal Z}$. By Definition~\ref{defn:partn-scn}, 
we have $|Y| \leq \delta m$. 
if the underlying scenario is in $Y$, then
$\NA$ terminates at time $t \leq \beta L (i+1)$; that is, $\NA$ does not go beyond phase $i$. So, we have 
\begin{equation}\label{eq:small-parts}
  \sum_{Y\in \mathcal{H(S)}}\pr(Y) \leq 1 - a(i+1)  
\end{equation}
Combining \eqref{eq:bad-parts} and \eqref{eq:small-parts} yields
\begin{align*}
\sum_{Z:\text{okay or good}} \pr(Z) &= 1 - \sum_{Z: \text{bad}} \pr(Z)  - \sum_{Y\in H(S)\setminus {\cal Z}}\pr(Y) \\
    &\geq 1 - 2\cdot o(i+1) - (1-a(i+1))\\
    & = a(i+1) - 2\cdot o(i+1).
\end{align*}
\end{proof}

\begin{lemma} For any $S\sse X$, the function 
$$g(T) =\sum_{Z \in {\cal Z}}\left( \sum_{\scn \in L_T(Z)} p_{\scn} + \sum_{\scn \in Z} p_{\scn} \cdot \frac{f(\psi_Z(S) \cup \psi_{\scn}(T)) - f(\psi_Z(S))}{Q - f(\psi_Z(S))} \right), \quad \forall T\sse X, $$ 
is nonnegative, monotone and submodular. Hence, the problem in \eqref{eqn:greedy-choice} is an instance of \RSO. 
\end{lemma}
\begin{proof}
We show that the term in function $g(T)$ for each $Z \in {\cal Z}$ is nonnegative, monotone and submodular. This suffices to prove the lemma as the sum of submodular functions remains submodular.  

We now fix $Z\in {\cal Z}$. 
Notice that $\sum_{\scn \in L_T(Z)} p_{\scn}$ as a function of $T$ is a weighted set-coverage function, which is known to be monotone submodular. 
Further, for any fixed $\scn\in Z$, the function $g_\scn(T):=f(\psi_Z(S) \cup \psi_{\scn}(T)) - f(\psi_Z(S))$ is nonnegative, monotone and submodular in $T$ because the original function $f$ is. This implies that $\sum_{\scn \in Z} p_{\scn} \cdot \frac{f(\psi_Z(S) \cup \psi_{\scn}(T)) - f(\psi_Z(S))}{Q - f(\psi_Z(S))}$, which is a nonnegative combination of $g_\scn(T)$s, is also  nonnegative, monotone and submodular.
\end{proof}

The following lemma gives a lower bound on the score of tour $\Pi_t$.

\begin{lemma}\label{lem:submod}
For time $t \in [\beta L i, \ \beta L (i+1))$, we have
$$\text{score}(\Pi_t) \geq \frac{\delta}{4\rho \cdot (i+1)} \cdot \left(a(i+1) - 2\cdot o(i+1)\right)$$
\end{lemma}
\begin{proof}
Recall that by construction of tour $T_Z$ that  $d(T_Z) \leq 2 \cdot (i+1)$. Let $T$ be a concatenation of all  tours $\{T_Z: Z \in \mathcal{Z} \}$. So, we have $d(T) = \sum_{Z \in \mathcal{Z}} d( T_Z) \leq |{\cal Z}| \cdot 2(i+1) \leq \frac{2\cdot (i+1)}{\delta}$, 
where the final inequality uses the fact that  $|\mathcal{Z}| \leq 1/\delta$. Moreover, tour $\Pi_t$ is a $\rho$-approximately optimal solution to the \RSO instance solved in \eqref{eqn:greedy-choice}; see Lemma~\ref{lem:submod}. So, we can lower  bound the score of $\Pi_t$ as follows.
\begin{align}
    \text{score}(\Pi_t) &\geq \frac{1}{\rho} \cdot \max_{\text{tour }\Pi } \text{score}(\Pi) \geq \frac{1}{\rho} \cdot  \text{\score}(T) \notag \\
    &= \frac{1}{\rho} \cdot \frac{1}{d(T)}\cdot \sum_{Z \in \mathcal{Z}} \left( \pr(L_{T}(Z)) + \sum_{\scn \in Z} p_{\scn} \cdot \frac{f(\psi_Z(S) \cup \psi_{\scn}(T)) - f(\psi_Z(S))}{Q - f(\psi_Z(S))} \right) \label{eq:avg-scores} \\
    &\geq \frac{1}{\rho} \cdot \frac{1}{d(T)}\cdot \left(\sum_{Z: \text{ okay}} \pr(L_{T}(Z)) + \sum_{Z: \text{ good}} p_{\scn} \cdot \frac{f(\psi_Z(S) \cup \psi_{\scn}(T)) - f(\psi_Z(S))}{Q - f(\psi_Z(S))}   \right) \notag \\
    &\geq \frac{1}{\rho} \cdot \frac{1}{d(T)}\cdot \left(\sum_{Z: \text{ okay}} \frac{\pr(Z)}{2} + \sum_{Z: \text{ good}} \frac{\pr(Z)}{2} \right) \label{eq:lower-bound-score} \\
    &\geq \frac{1}{\rho} \cdot \frac{1}{d(T)}\cdot \frac{a(i+1) - 2\cdot o(i+1)}{2} \label{eq:lower-bound-non-completion}\\
    &\geq \frac{1}{\rho} \cdot \frac{\delta}{4\cdot (i+1)} \cdot \left(a(i+1) - 2\cdot o(i+1)\right) \label{eq:tour-combined-cost}
\end{align}
where \eqref{eq:avg-scores} uses the definition of score , \eqref{eq:lower-bound-score} uses Lemmas~\ref{lem:info_gain} and \ref{lem:stem_func_gain} and $T_Z\sse  T $ for all $Z\in {\cal Z}$, \eqref{eq:lower-bound-non-completion} uses Lemma~\ref{lem:noncompletion_prob_bound}, and \eqref{eq:tour-combined-cost} uses the upper bound on the distance of tour $T$.

\end{proof}

On summing over $t \in [\beta Li, \ \beta L (i+1))$, we get
\begin{equation}\label{eq:G-i-lb}
    G_i = \sum_{t=\beta Li}^{\beta L (i+1)} \score(\Pi_t) \geq \frac{\beta L \delta}{4\rho\cdot (i+1)}\cdot \left(a(i+1) - 2\cdot o(i+1)\right)
\end{equation}

\paragraph{Upper bounding $G_i$.} Instead of upper bounding the gain term $G_i$, we show a stronger result by bounding the sum $\sum_{j \geq i} G_j$. Specifically, we prove the following.
\begin{equation}\label{eq:G-i-ub}
    \sum_{j \geq i} G_j \leq L \cdot a(i).
\end{equation}
The proof of this proceeds by viewing $G_i$ as a sum over decision paths.
Towards this end, fix scenario $\scn \in M$. Let $G_i(\scn)$ denote the gain term \emph{conditioned} on the underlying scenario being $\scn$.
We will provide an upper bound for the term $\sum_{j \geq i} G_j(\scn)$.
Let $\langle \Pi_1, \Pi_2, \ldots \rangle$ denote the sequence of tours given by $\NA$. 
Let $\Pi_h$ be the first tour (if any) selected in phase $i$, and let $\Pi_{\ell}$ denote the final tour selected (before $\NA$ terminates). 
We set $G_j(\scn) = 0$ for all $j \geq i$, if $h$ is undefined.
Let $Z_j$ be the set of scenarios compatible with the realization $\cup_{p=1}^{j-1}\psi_{\scn}(\Pi_p)$. So, $M = Z_1 \supseteq Z_2 \cdots Z_{\ell} \supseteq Z_{\ell+1}$. Here $Z_{\ell+1}$ denotes the set of scenarios that are compatible when $\NA$ terminates. 
Lastly, let $c_p$ denote the cost of tour $\Pi_p$ paid \emph{after} phase $i$: note that $c_p = d(\Pi_p)$ for all $p > h$. So, we have 

\begin{align}
    \sum_{j \geq i} G_j(\scn) &= \sum_{p=h}^{\ell} \frac{c_p}{d(\Pi_p)}  \left(\mathbbm{1} \left[\scn \in L_{\Pi_{p}}(Z_p) \right] + \frac{f\left(\psi_{\scn}(\Pi_p) \cup \cup_{q=1}^{p-1}\psi_{\scn}(\Pi_q) \right) - f\left(\cup_{q=1}^{p-1}\psi_{\scn}(\Pi_q)\right)}{Q - f\left(\cup_{q=1}^{p-1}\psi_{\scn}(\Pi_q)\right)}\right) \notag \\
    &\leq \sum_{p=h}^{\ell} \left(\mathbbm{1} \left[\scn \in L_{\Pi_{p}}(Z_p) \right] + \frac{f\left(\psi_{\scn}(\Pi_p) \cup \cup_{q=1}^{p-1}\psi_{\scn}(\Pi_q) \right) - f\left(\cup_{q=1}^{p-1}\psi_{\scn}(\Pi_q)\right)}{Q - f\left(\cup_{q=1}^{p-1}\psi_{\scn}(\Pi_q)\right)}\right). \label{eq:sum-gain-i-ub}
\end{align}
where the inequality follows from \(c_p \leq d(\Pi_p)\). Note that $\sum_{j \geq i}G_j = \E_{\scn}\left[\sum_{j \geq i}G_j(\scn)\right]$.
We recall that that $\NA$ 
continues visiting sensing locations
until either
(i) the number of compatible scenarios is less than $\delta m$, or (ii) the function is already covered (i.e., $f(S) = Q$). In particular, this implies that $|Z_{\ell}| \geq \delta m$.
We proceed by separately bounding the terms in \eqref{eq:sum-gain-i-ub}. We begin by analyzing the second term. Since $f$ is integral, monotone and has full coverage value equal to $Q$, we get
\begin{equation}\notag
      \mathlarger{\sum}_{p=h}^{\ell} \left( \frac{f\left(\psi_{\scn}(\Pi_p) \cup \cup_{q=1}^{p-1}\psi_{\scn}(\Pi_q) \right) - f\left(\cup_{q=1}^{p-1}\psi_{\scn}(\Pi_q)\right)}{Q - f\left(\cup_{q=1}^{p-1}\psi_{\scn}(\Pi_q)\right)}\right) \leq \sum_{p=1}^{Q}\frac{1}{p} \leq \log(Q).
\end{equation}
Next, we analyze the first term: $\sum_{p=h}^{\ell} \mathbbm{1} \left[\scn \in L_{\Pi_{p}}(Z_p) \right]$. 
Consider $h \leq p \leq \ell$. If $\scn \in L_{\Pi_p(Z_p)}$, then the number of compatible scenarios is at most half the total number of scenarios in $Z_p$ (this follows from Definition~\ref{defn:parts-item}). So $|Z_{p+1}| \leq \frac{1}{2}\cdot |Z_p|$.
Since 
$|Z_1| = m$, and $|Z_{\ell}| \geq \delta m$,
it follows that

\[\sum_{p=h}^{\ell}\mathbbm{1} \left[\scn \in L_{\Pi_p}(Z_p) \right] \leq \log_2\left(\frac{m}{\delta m}\right) = \log\left(\frac{1}{\delta}\right).\]

Thus, we have $\sum_{j \geq i}G_j(\scn) = \log\left(\frac{1}{\delta}\right) + \log(Q) = \log\left(\frac{Q}{\delta}\right)$.
Finally, we have 
\begin{equation}\notag
\sum_{j \geq i}G_j = \E_{\scn}\left[\sum_{j \geq i}G_j(\scn)\right] \leq \log\left(\frac{Q}{\delta}\right) \cdot \pr(\NA \text{ does not terminate before phase $i$}) = \log\left(\frac{Q}{\delta}\right) \cdot a(i),
\end{equation}
which proves \eqref{eq:G-i-ub} with $L = \log\left(\frac{Q}{\delta}\right)$.

\paragraph{Finishing Up.} On summing \eqref{eq:G-i-lb} for $j \geq i$, and using \eqref{eq:G-i-ub} we get 
\begin{equation}\notag
\frac{\beta L \delta}{4\rho} \cdot \mathlarger{\sum}_{j \geq i}\left(\frac{a(j+1) - 2\cdot o(j+1))}{j+1} \right)\quad \leq \quad \sum_{j \geq i} G_j \quad \leq \quad L \cdot a(i),
\end{equation}
which on rearranging gives
\begin{equation}\notag
    \frac{\beta\delta}{4\rho}  \cdot \mathlarger{\sum}_{j \geq i} \left(\frac{a(j+1) - 2\cdot o(j+1)}{j+1}\right) \leq a(i).
\end{equation}
This completes the proof of Lemma~\ref{lem:key-lemma}.

\section{Improved Algorithm for Hypothesis Identification}\label{sec:ipp-h}
\vspace{-2mm}
Here, we consider an important special case of \prob: path planning for hypothesis identification (\probh). An instance is given by  the tuple $(X, r, d, M, \D, O)$. Here, $X$ is the set of 
sensing locations,
$r$ is the root location and $d$ is a metric on $X \cup \{r\}$.
 Set $M = \{1, \ldots, m\}$ is a  finite set of \emph{hypotheses}/\emph{scenarios} with probabilities  $\{p_\scn\}_{\scn\in M}$.
The set $O$ denotes the  possible observations;  each $v  \in X$ realizes to a random observation in 
$O$. 

The distribution $\D$ specifies the probability $p_\scn$ of each scenario $\scn\in M$ as well as the observations $\{\scn(v) \in O : v\in X\}$ at all locations under scenario $\scn$. The true hypothesis/scenario  $\scn^*\in M$ according to the distribution $\D$; however $\scn^*$ is initially unknown to the algorithm. The goal is to identify $\scn^*$ by visiting locations at the minimum expected distance. When location $v\in X$ is visited, the robot observes $\scn^*(v)\in O$ and can use this information to update its priors.

In order to cast \probh as a special case of \prob, we define a submodular function $f:2^{X\times O}\rightarrow \R_+$ as follows. 
For each $v\in X$ and $o\in O$, let $E_{v,o}\sse M$ be the set of hypotheses that are {\em incompatible} with observation $o$ at location $v$. Note that if we observe $o$ at location $v$ then we must have $\scn^*\not\in E_{v,o}$. Now, we define, 
$\forall S\sse X \mbox{ and } o_v\in O \mbox{ for }v\in S$:
\begin{equation}\label{eq:f-hyp-id}
f(\{(v,o_v) : v\in S\}) = \big| \bigcup_{v\in S} E_{v,o_v} \big|. 
\end{equation} 
Note that this is exactly the number of incompatible hypotheses after having visited locations $S$ and observed $o_v$ at each $v\in S$. It is easy to see that $f$ is monotone and submodular: it is a set coverage function. 
Clearly, $\scn^*$ is identified  precisely when this number is $m-1$. So, we set our target $Q=m-1$. 

Recall that at each step of the partial covering algorithm $\mathtt{PCA}$, we need to solve an instance of \RSO with the following objective function:
\begin{align*}
 g(T) &=  \sum_{Z \in {\cal Z}}\bigg( \sum_{\scn \in L_T(Z)} p_{\scn} + \sum_{\scn \in Z} p_{\scn} \cdot \frac{f(\psi_Z(S) \cup \psi_{\scn}(T)) - f(\psi_Z(S))}{Q - f(\psi_Z(S))} \bigg),
\end{align*}
$\forall T\sse X$. Above, $S\sse X$ is a fixed subset. 
This is exactly the criterion in \eqref{eqn:greedy-choice}.

In the special case of \probh, we show below that all \RSO instances correspond to the simpler {\em ratio group Steiner} problem: 
\begin{definition}[Ratio Group Steiner] \label{defn:ratio-grp}
An instance consists of a metric $(V,d)$ with nodes $V$ and distances $d:V\times V \to \R_+$. There is a special root node $r\in V$ and $k$ {\em groups}, where each group $i\in [k]$ is associated with a subset $S_i\sse V$ and weight $w_i\ge 1$. We want to find a tour $\tau$ originating from $r$ that minimizes the ratio of its distance to the weight of groups covered.
\end{definition}

Moreover,  we provide a better and more efficient approximation algorithm for this problem (see \S~\ref{sec:ratio-group-steiner}).
\begin{theorem}\label{thm:grp-st-ratio}
There is a randomized $O(\log^2n)$-approximation algorithm for ratio group Steiner, where $n=|V|$ is the number of nodes. 
\end{theorem}

We now construct an instance of ratio group Steiner corresponding to the \RSO instance when function $f$ is given by \eqref{eq:f-hyp-id}. The metric and root remains the same. The groups and weights are as follows:
\vspace{-2mm}
\begin{itemize}
\item {\em Groups for information gain (1st term in $g$).} For each $Z\in {\cal Z}$ and scenario $\scn\in Z$ there is a group 
consisting of nodes $\{v\in X: \scn\in L_v(X)\}$ with weight $p_\scn$.  
\item {\em Groups for function gain (2nd term in $g$).} For each $Z\in {\cal Z}$, note that the compatible scenarios after observing partial realization $\psi_Z(S)$ is exactly $Z$. 
So, $f(\psi_Z(S))=m-|Z|$ and $\forall R\sse X\times O$,
$$f(\psi_Z(S)\cup R)-f(\psi_Z(S))=|Z\cap \left( \cup_{(v,o)\in R} E_{v,o}\right)|.$$
Hence, for any $\scn\in Z$ and $\forall T\sse X$, we have
\begin{align*}
   \frac{f(\psi_Z(S) \cup \psi_{\scn}(T)) - f(\psi_Z(S))}{Q - f(\psi_Z(S))} = \frac{1}{|Z|-1}\sum_{\theta\in Z} \mathbf{1}\left[\theta \in \cup_{v \in T} E_{v,\psi_\scn(v)}\right].
\end{align*}
Moreover, note that $\theta \in \cup_{v \in T} E_{v,\psi_\scn(v)}$ iff the outcomes at $v$ under scenarios $\scn$ and $\theta$ are different, i.e., $\psi_\scn(v)\ne \psi_\theta(v)$.
Now, we introduce a group for each $Z\in {\cal Z}$ and scenarios $\scn,\theta\in Z$ with nodes $\{v\in X : \psi_\scn(v)\ne \psi_\theta(v)\}$ and weight $\frac{p_\scn}{|Z|-1}$.
\end{itemize}
\vspace{-2mm}
So, it follows that for any $T\sse X$, the total weight of covered groups is $g(T)$. Hence, this \RSO instance reduces to the ratio group Steiner problem.

\paragraph{Completing the proof of Theorem~\ref{thm:k-round-id}.} Using the above ratio group Steiner instance and Theorem~\ref{thm:grp-st-ratio}, we obtain a  $\rho=\bigOh(\log^2n)$ approximation algorithm for \RSO instances arising from \probh. Combined with Theorem~\ref{thm:k-round-ipp}, this  implies Theorem~\ref{thm:k-round-id}.

\vspace{-2mm}
\paragraph{Tighter Approximation Using More Rounds.} Using our  partial covering algorithm (Theorem~\ref{thm:partial-cover-thm}) and a  different  measure of progress in each round  (as in Theorem 6.7 of \cite{GGN21}), we can also obtain $2k$-round algorithms with better approximation guarantees of $\bigOh\left(\log^{2+\epsilon}(n) \cdot m^{1/k} \cdot \log(Q m)\right)$ for \prob and $\bigOh\left(\log^{2}(n) \cdot m^{1/k}\cdot \log m)\right)$ 
for \probh \ {(see \S~\ref{app:better-appx} for details)}. Setting the number of rounds to $\bigOh(\log m)$, we then 
 get  approximation guarantees of $\bigOh\left(\log^{2+\epsilon}(n) \cdot \log(Q m)\right)$ and $\bigOh\left(\log^2(n) \cdot \log m\right)$ for \prob and \probh respectively. These approximation ratios match the previous-best approximation ratios for these problems, even for fully-adaptive algorithms \cite{NavidiKN20,GNR17}. In fact, \probh and \prob generalize the {\em group Steiner tree} problem \cite{GargKR00}, for which the best known approximation ratio is $\bigOh\left(\log^2(n) \cdot \log m\right)$; there is also an $\Omega(\log^{2-\epsilon} n)$ hardness of approximation \cite{HalperinK03}. 

\vspace{-2mm}

\subsection{Ratio  Group Steiner Tree}\label{sec:ratio-group-steiner}
\def\cov{\mathsf{cov}}
\def\gkr{\ensuremath{\mathsf{GKR}}\xspace}
In this section, we consider the minimum ratio group Steiner  problem (Definition~\ref{defn:ratio-grp}). 
Recall that we have a metric $(V,d)$ with nodes $V$ and distances $d:V\times V \to \R_+$, root node $r\in V$ and $k$ {\em groups} where each group $i\in [k]$ has a subset $S_i\sse V$ and weight $w_i\ge 1$. We want to find a tour $\tau$ originating from $r$ that minimizes the ratio of its distance to the weight of groups covered. Formally, a group $i\in [k]$ is covered by tour $\tau$ if it contains any node of $S_i$ (i.e., $S_i\cap \tau \ne \emptyset$). Then, we want to find a tour $\tau$ that minimizes
$$Ratio(\tau)=\frac{d(\tau)}{\cov(\tau)}=\frac{\sum_{e\in \tau} d_e}{\sum_{i : S_i\cap \tau \ne \emptyset}w_i},$$
where  $d(\tau)$ is the total distance on tour $\tau$ and $\cov(\tau)$  is  its total coverage.  In this section, we prove Theorem~\ref{thm:grp-st-ratio}, which we restate for completeness.
\begin{theorem}\label{thm:grp-st-ratio-app}
There is a randomized $\bigOh(\log^2n)$-approximation algorithm for ratio group Steiner, where $n=|V|$ is the number of nodes. 
\end{theorem}
We note that the same approximation ratio was obtained previously  in \cite{CharikarCGG98}. However, our algorithm is   simpler because it involves a simpler ``pessimistic estimator''. Specifically, the algorithm in \cite{CharikarCGG98} is a de-randomization of the original group Steiner  rounding algorithm/analysis from \cite{GargKR00}. Whereas, our algorithm is a  de-randomization of a simpler algorithm/analysis. 

For completeness, we present all the details.

As is standard in group Steiner algorithms, we make use of the following well-known tree embedding result. 
\begin{theorem}[\cite{FakcharoenpholRT03}]\label{thm:frt}
Given any metric $(V,d)$, there is a polynomial time algorithm that returns a random tree $G=(V, E)$ with  edge-lengths  $\ell: E \to \R_+$ such that for every pair $u,v\in V$, we have:
\begin{itemize}
\item $\ell_G(u,v)\ge d(u,v)$ with probability one, and 
\item the expectation $\E\left[\ell_G(u,v)\right] \le \bigOh(\log |V|)\cdot d(u,v)$.
\end{itemize} 
Above, $\ell_G(u,v)$ is the distance between $u$ and $v$ on the tree, obtained by adding the lengths of all edges on the $u-v$ path in $G$. Furthermore, the depth of tree $G$ is bounded by $\bigOh(\log |V|)$.
\end{theorem}

Henceforth, we assume that the metric is given by a tree  $G = (V, E)$ rooted  at $r$ with  edge-lengths  $d: E \to \R_+$. Note that tours in a tree metric are equivalent to subtrees (where each edge in the subtree is traversed twice). So, we focus on finding a subtree $T$ containing $r$ that minimizes $ratio(T)$. By duplicating nodes (and adding zero length edges), we can assume (without loss of generality) that the groups are disjoint, i.e., $S_i\cap S_j=\emptyset$ for all $i\ne j$. Moreover, we can assume (w.l.o.g.) that no pair of nodes in the same group share an ancestor-descendant relation: if this were the case, we can remove the descendant node from the group (as the ancestor node must be visited if the descendant is).

 We also let $H$ denote the depth of the tree, i.e., the maximum number of edges on any path from $r$ to a leaf. By Theorem~\ref{thm:frt}, we know that $H=\bigOh(\log n)$. 

We use  a linear program (LP)  relaxation to the ratio problem.  We first set up some notation. 
We  assign directions to  edges    so that every edge is directed away from the root $r$.  For any edge $e\in E$, let $\pi(e)$ denote the parent edge of $e$ (if any). We overload this notation and also use $\pi(v)$ for any node $v$ to denote its parent edge. 

\begin{align}
\min&\quad \sum_{e\in E} d_e \cdot x_e& & &\\
s.t. & \quad x_e  \leq  x_{\pi(e)} & \forall e\in E \label{LP:mon-tree} \\
& \quad f^i_{e}  = \sum_{e':\pi(e') = e} f_{e'}^i & \forall e \in E \setminus\{\pi(v): v\in S_i\}\quad \forall i\in [k] \label{LP:flow}\\
&\quad \sum_{v\in S_i} f_{\pi(v)}^i =  y_i& \forall i\in[k] \label{LP:coverage}\\
& \quad f_e^i  \leq  x_e & \forall e\in E \quad\forall i\in[k] \label{LP:capacity}\\
& \quad \sum_{i = 1}^k  w_i\cdot  y_i  \geq  1\label{LP:normalize}\\
&\quad  x, f, y  \geq  0. &   
\end{align}

In this LP, variables $x$ correspond to the edges in the solution subtree, variables $y$ indicate the groups that get covered, and for each group $i$, variables $f^i$ establish a flow from $r$ to $S_i$. Note that constraint \eqref{LP:flow} is the flow balance constraint, which ensures that each $f^i$ is a flow from $r$ to $S_i$. We first show that this LP is indeed a relaxation.
\begin{lemma}
The optimal value of the above LP is at most the optimal value of the ratio group Steiner problem.
\end{lemma}
\begin{proof}
Consider the optimal (integral) solution to the ratio group Steiner problem, given by subtree $T^*$ containing $r$. Let $W=\cov(T^*)$ denote the total weight of groups covered by $T^*$. Note that $W\ge 1$ as each weight is at least $1$. Then, the optimal value is $\OPT=\frac{d(T^*)}{W}$. 

We define the following fractional solution corresponding to $T^*$.
$$x_e =\left\{ \begin{array}{ll}
1/W & \mbox{ if }e\in E(T^*)\\
0 & \mbox{ otherwise}
\end{array}
\right., \quad \forall e\in E.
$$
$$y_i =\left\{ \begin{array}{ll}
1/W & \mbox{ if group $i$ covered by } T^*\\
0 & \mbox{ otherwise}
\end{array}
\right., \quad \forall i\in [k].
$$
For each group $i$ that is covered by $T^*$, choose some vertex $v_i\in S_i\cap T^*$, and define:
$$f^i_e =\left\{ \begin{array}{ll}
1/W & \mbox{ if $e$ is on the $r-v_i$ path}\\
0 & \mbox{ otherwise}
\end{array}
\right., \quad \forall e\in E.
$$
For groups $i$ that are not covered by $T^*$, we set $f^i=0$.

Constraint \eqref{LP:mon-tree} is satisfied because $T^*$ is a subtree containing $r$: so if an  edge $e\in E(T^*)$ then its parent $\pi(e)\in E(T^*)$ as well. 
Constraint \eqref{LP:flow} is clearly satisfied for groups that are not covered by $T^*$. Constraint \eqref{LP:flow} is also satisfied for groups $i$ that are covered by $T^*$ because $f^i$ corresponds to a path from $r$ to $v_i\in S_i$. Constraints \eqref{LP:coverage} and \eqref{LP:capacity} are satisfied by definition of the variables. Finally, constraint \eqref{LP:normalize} is satisfied by definition of $W$. 
\end{proof}

We can assume (without loss of generality) that any LP solution $(x,y,f)$ satisfies:
$$f_v^i = x_{\pi(v)} \quad \forall v \in S_i \quad \forall i \in [k].$$
To see this, suppose that we have $f_v^i < x_{\pi(v)} $  for some $v\in S_i$ and group $i$. Then,  we can decrease $x_{\pi(v)}$ to $f^i_v$ without violating any other constraints.  This uses the fact that all groups are disjoint: so edge $\pi(v)$ is only used to cover group $i$. Note that this change only reduces the objective value. 

Given the $x$ variables, we can infer the $y$ and $f$ variables by solving a max-flow instance for each group $i$ (that sends the maximum flow from $r$ to $S_i$). So,  we often refer to the fractional solution $(x,y,f)$ as just $x$.

In this section, we will show:

\begin{theorem}\label{thm:tree-ratio}
There is a rounding algorithm for the ratio group Steiner   	problem on trees that finds a solution of ratio at most $H+1$ times the optimal LP value.

\end{theorem}
Combined with the tree embedding result (Theorem~\ref{thm:frt}), this implies Theorem~\ref{thm:grp-st-ratio-app}.

\subsubsection{Basic Randomized Rounding}
We first present a basic random rounding procedure for group Steiner tree, due to \cite{GargKR00}. We refer to this procedure as \gkr rounding.

\begin{algorithm}[t]
\caption{Group Steiner rounding algorithm $\gkr(x)$} \label{alg:GKR}
\begin{algorithmic}[1]
  \State For each edge $e\in E$, select it into $E'$ independently with probability $\frac{x_e}{x_{\pi(e)}}$.
  \State Let $T'\sse E'$ denote the component (subtree) containing the root $r$.
  \State   Return subtree $T'$.
\end{algorithmic}
\end{algorithm}

We now state some properties of this rounding algorithm.

\begin{lemma}\label{lem:gkr-1}
For any fractional solution $x$, the expected cost of the subtree $T'=\gkr(x)$ is $\sum_{e\in E} d_e x_e$. Moreover, for any group $i\in [k]$ and node $v\in S_i$, we have $\Pr[v \in T'] = x_{\pi(v)} = f_v^i$.
\end{lemma}
\begin{proof}
It is easy to see that $\Pr[e\in T'] = x_e$ for each edge $e\in E$. This implies the first statement directly. For the second statement, we also use the fact that $x_{\pi(v)} = f_v^i$ for any $v\in S_i$. 
\end{proof}

\begin{lemma}For any group $i\in [k]$ let $T_i = T'\cap S_i$ be the nodes in $S_i$ that are covered by the subtree $T'=\gkr(x)$. Then, 
    $\mathbb{E}\left[|T_i| \,\middle| \, v\in T_i \right]\leq H+1$ for any $v\in S_i$ and $i\in [k]$.
    \label{lem:gkr-2}
\end{lemma}
\begin{proof}
Fix any group $i$ and node $v\in S_i$. Let $\langle r=u_0, u_1, \cdots u_t=v\rangle$ denote the $r-v$ path in $G$. Note that $t\le H$. For each $j=0,1,\cdots t-1$ let $L_j$ denote the subtree rooted at $u_j$ containing all nodes that are {\em not} descendants of   
$u_{j+1}$. Note that $T_i\setminus \{v\} = \cup_{j=0}^{t-1} \left(L_j \cap T_i\right)$ because no $S_i$-node can be a descendant of $v$. 

Now, consider any $j$ and let edge $e_j=(u_{j-1},u_j)$. (If $j=0$ then $e_j$ is undefined and we set $x_{e_j}=1$ below.) We now bound
$$\E\left[ |L_j \cap T_i| \,\mid\, v\in T_i \right] = \sum_{w\in L_j} \Pr\left[ w\in T_i \,|\, v\in T_i\right] = \sum_{w\in L_j} \frac{f^i_w }{x_{e_j}} \le \frac{f^i_{e_j}}{x_{e_j}}\le 1.$$
The second equality above uses the fact that conditioned on $v\in T_i$ (i.e., all edges on the $r-v$ path are selected in $\gkr(x)$), for $w$ to be in $T_i$ we just need all the edges on the $u_j-w$ path to be selected: the joint probability of this event is $\frac{f^i_w }{x_{e_j}}$. The first inequality above uses the flow-conservation constraints \eqref{LP:flow} for $f^i$, which implies that all the $\{f^i_w : w\in L_j\}$ flow passes through edge $e_j$. The last inequality is by constraint \eqref{LP:capacity}.

Adding over all $j=0,1,\cdots t-1$ and including $v$, we get 
$$\E \left[|T_i| \,\middle| \, v\in T_i \right] = 1+ E \left[|T_i\setminus v| \,\middle| \, v\in T_i \right] =1+\sum_{j=0}^{t-1}  \E\left[ |L_j \cap T_i| \,\mid\, v\in T_i \right]  \leq H+1. $$
This completes the proof.
\end{proof}

\subsubsection{Deterministic Rounding}
For any fractional solution $x$, we define estimates of its distance and coverage as follows:
$$D({x}) = \sum_{e\in E}  d_e x_e, \qquad \mbox{and}$$
\begin{align*}
P (x)  & = \sum_{i=1}^k w_i\cdot \sum_{v\in S_i} \left[\Pr\left[v\in \gkr(x)\right] - \frac{1}{2H+2}\sum_{u\in S_i}\Pr \left[v,u \in \gkr(x)\right]\right] \\
& = 
\sum_{i=1}^k w_i\cdot \sum_{v\in S_i} \left[ x_{\pi(v)}  - \frac{1}{2H+2}\cdot \sum_{u\in S_i} \frac{x_{\pi(v)} \cdot x_{\pi(u)}}{x_{a(u,v)}} \right] .\end{align*}
Above, for any pair of nodes $u,v$, we use $a(u,v)\in E$ to denote their least common ancestor  edge. If there is no such edge (i.e., the $r-u$ and $r-v$ paths have no common edge) then we set $x_{a(u,v)}=1$. $D(x)$ is clearly the expected length of solution $\gkr(x)$. The definition of $P(x)$ is based on  a pessimistic estimator of the coverage of   $\gkr(x)$.

The deterministic rounding algorithm starts off with an optimal LP solution and modifies it gradually. It considers edges in  increasing order of depth. When edge $e\in E$ is considered, the algorithm modifies the current fractional solution $x$ so that $x_e=0$ or $x_e=1$ (other variables in the subtree below $e$ will also change). At the end of the algorithm, all edges will have integral values, which means we get an integral solution. 

When processing edge $e$, let $x$ be the current fractional solution. Let $G_e$ denote the subtree below (and including) $e$. We define two new fractional solutions as follows:
$$x'_g = \left\{\begin{array}{ll}
0 & \mbox{ if }g\in G_e\\
x_g& \mbox{ otherwise}
\end{array}\right., \qquad \forall g\in E$$
$$x''_g = \left\{\begin{array}{ll}
\frac{x_g}{x_e} & \mbox{ if }g\in G_e\\
x_g& \mbox{ otherwise}
\end{array}\right., \qquad \forall g\in E$$
Note that $x'_e=0$ and $x''_e=1$. We update the fractional solution to either $x'$ or $x''$ depending on which has the smaller ratio. That is, $x^{new}\gets  \arg\min \left\{\frac{D(x')}{P(x')},\frac{D(x'')}{P(x'')} \right\}$.

A crucial property of this rounding algorithm is that the ratio of the current fractional solution never increases, i.e.,
\begin{equation}\label{eq:round-ratio}
\min \left\{\frac{D(x')}{P(x')},\frac{D(x'')}{P(x'')} \right\}\leq\frac{D( x )}{P( x )}.
\end{equation}

We first complete the proof of Theorem~\ref{thm:tree-ratio} assuming this property. Let $\hat{x}$ be the optimal LP solution (which is our starting solution) and $x^*$ be the rounded integral solution (which is the final solution).  

\begin{lemma}
    We have  $P(x^*) \leq \frac{H+1}{2}\cdot \cov(x^*)$. 
    \label{lem:profit_coverage_bound}
\end{lemma}
\begin{proof}
Recall that $x^*$ is an integral solution and $\cov(x^*)$ is the total weight of groups covered by it. 
  Fix any group $i\in [k]$, and suppose that $r$ nodes of $S_i$ are connected to $r$ in solution $x^*$. 
  
  If $r =0$ then the contribution of group $i$ to  both $P(x^*)$ and  $\cov(x^*)$ is zero. 

Now, suppose that $r\geq 1$. Then, group $i$ contributes $w_i$ to $\cov(x^*)$. Group $i$'s contribution to $P(x^*)$ is
$$ w_i\cdot \sum_{v\in S_i} \left[ x^*_{\pi(v)}  - \frac{1}{2H+2}\cdot \sum_{u\in S_i} \frac{x^*_{\pi(v)} \cdot x^*_{\pi(u)}}{x^*_{a(u,v)}} \right] = w_i\left(r-\frac{r^2}{2H+2}\right),$$
where we used the fact that $x^*$ is integral and there are $r^2$ pairs of nodes $u,v\in S_i$ with $x^*_{\pi(u)}>0$ and $x^*_{\pi(v)}>0$.   It is easy to see that the maximum (over $r$) is $w_i\cdot \frac{H+1}2$ obtained when $r=H+1$. 

Thus, we get that the contribution of any group $i$ to $P(x^*)$ is at most $\frac{H+1}2$ times its contribution to $\cov(x^*)$. 
    Summing over all groups yields the lemma.
\end{proof}

\begin{lemma} If $(\hat{x},y,f)$ is the optimal LP solution then $P(\hat{x}) \geq \frac{1}{2}\sum_{i=1}^k w_i \cdot y_i$.
 \label{lem:profit_flow_bound}
\end{lemma}
\begin{proof}
Let $T'=\gkr(\hat{x})$.   
Fix any group $i$ and node $v\in S_i$. Let 
\begin{align*}
    P_{i,v} & :=  \Pr\left[v\in T'\right] - \frac{1}{2H+2} \sum_{u\in S_i} \Pr\left[u\in T'\mid v\in T'\right]\Pr\left[v\in T'\right] \\
    &= \Pr\left[v \in T'\right] - \frac{\Pr[v\in S']}{2H+2}\cdot \mathbb{E}\left[|T'\cap S_i|\mid v\in T' \right]\\
    &\geq \Pr[v\in T'] - \frac{1}{2}\Pr\left[v\in T'\right]  \quad = \quad \frac{1}{2}\Pr\left[v\in T'\right] \quad = \quad \frac12\cdot f^i_{\pi(v)} ,
\end{align*}
where the inequality uses Lemma~\ref{lem:gkr-2} and the last equality uses Lemma~\ref{lem:gkr-1}.

Adding over all $v\in S_i$, we have $\sum_{v\in S_i} P_{iv} \ge \frac12 \sum_{v\in S_i} f^i_{\pi(v)} = \frac{y_i}{2}$.

Now, adding over groups $i$,  we get     
$$P(\hat{x}) = \sum_{i=1}^k w_i \sum_{v\in S_i} P_{iv} \geq \frac{1}{2}\sum_{i=1}^k w_i \cdot y_i,$$
which completes the proof.
\end{proof}

\paragraph{Completing proof of Theorem~\ref{thm:tree-ratio}.} We now show that our solution $x^*$ has ratio at most $H$ times that of the LP optimum. We have
 $$\text{Ratio}(x^*)= \frac{D(x^*)}{\cov(x^*)} \overset{(a)}{\leq} \frac{H+1}2\cdot  \frac{D(x^*)}{P(x^*)} \overset{(b)}{\leq} \frac{H+1}2\cdot  \frac{D(\hat{x})}{P(\hat{x})} \overset{(c)}\leq (H+1) \frac{\sum_{e\in E} d_e\cdot \hat{x}_e}{\sum_{i=1}^k w_i\cdot y_i} = (H+1) LP_{ratio},$$
where inequality (a) follows from Lemma \ref{lem:profit_coverage_bound}, inequality $(b)$ follows from \eqref{eq:round-ratio}, and (c) follows from Lemma \ref{lem:profit_flow_bound}.

\paragraph{Proof of Equation \eqref{eq:round-ratio}.} Recall that $x$ is the current fractional solution (just before processing edge $e$) and $x'$ and $x''$ are the two new solutions (with  $x'_e=0$ and $x''_e=1$). By definition of $x'$ and $x''$, it is clear that $D(x) = x_e\cdot D(x'') + (1-x_e)\cdot D(x')$. 

Next, we will show that $P(x) = x_e\cdot P(x'') + (1-x_e)\cdot P(x')$. To this end, let 
$$P_1(x)=
\sum_{i=1}^k w_i\cdot \sum_{v\in S_i} x_{\pi(v)}, \qquad \mbox{and}\qquad  P_2(x)=
\sum_{i=1}^k w_i\cdot  \sum_{u\in S_i} \frac{x_{\pi(v)} \cdot x_{\pi(u)}}{x_{a(u,v)}}  $$
correspond to the two parts of $P(x)$. So, $P(x) = P_1(x)-\frac{1}{2H}\cdot P_2(x)$. 

Again, it is easy to see that $P_1(x) = x_e\cdot P_1(x'') + (1-x_e)\cdot P_1(x')$ because $P_1()$ is a linear function. 
Below, we show that $P_2(x) = x_e\cdot P_2(x'') + (1-x_e)\cdot P_2(x')$, which would imply the corresponding equation for $P$. To see this, we fix group $i$ and nodes $v,u\in S_i$. The following table shows the contribution of $\langle i, v, u\rangle$ to $P_2(\cdot)$ in different cases:

\begin{table}[H]
    \centering
    \begin{tabular}{l|ccc}
     & in $P_2(x')$ & in $P_2(x'') $ & in  $(1-x_e) P_2(x') + x_e P_2(x'')$ \\\hline
     $u\in G_e, v\in G_e$& $0$ & $ \frac{x_{\pi(u)}}{x_{e}}\frac{x_{\pi(v)}/x_e}{x_a/x_e}$ & $ x_{\pi(u)}\frac{x_{\pi(u)}}{x_a}$ \\
     $u\in G_e, v\notin G_e$& $0$ & $ \frac{x_{\pi(u)}}{x_e}\frac{x_{\pi(v)}}{x_a}$ & $ x_{\pi(u)}\frac{x_{\pi(u)}}{x_a}$\\
     $u \notin G_e, v\in G_e$& $0$ & $ x_{\pi(u)}\frac{x_{\pi(v)}/x_e}{x_a}$ & $ x_{\pi(u)}\frac{x_{\pi(u)}}{x_a}$\\
     $u \notin G_e, v\notin G_e$& $x_{\pi(u)}\frac{x_{\pi(v)}}{x_a}$ &  $x_{\pi(u)}\frac{x_{\pi(v)}}{x_a}$ & $ x_{\pi(u)}\frac{x_{\pi(v)}}{x_a}$
    \end{tabular}
\end{table}
Above, $a=a(u,v)$ is the least common ancestor edge of $u$ and $v$. Note that in the first case above, $a\in G_e$ because both $u,v\in G_e$. In all other cases, $a\not\in G_e$.  It can be observed that, in all cases, the contribution of $\langle i, v, u\rangle$ to   $(1-x_e)P_2(x') + x_e P_2(x'')$ equals $\frac{x_{\pi(u)}\cdot x_{\pi(v)}}{x_a}$ which is its contribution to $P_2( x )$. Adding up the contributions over 
 $\langle i, v, u\rangle$ we get $P_2(x) = x_e\cdot P_2(x'') + (1-x_e)\cdot P_2(x')$, as desired.

Finally, we have
$$\min \left\{\frac{D(x')}{P(x')},\frac{D(x'')}{P(x'')} \right\} \leq\frac{x_e D(x'') + (1-x_e)D(x')}{x_e P(x'') + (1-x_e )P(x')} = \frac{D( {x})}{P( {x})}.$$ 
This proves equation \eqref{eq:round-ratio}.

\section{Computational Results}\label{app:experiments}
In this section, we provide a detailed summary of computational results of our $k$-round algorithm for  path planning for hypothesis identification (\probh).
We test our algorithm on two sets of instances: UAV search, (that were also used in~\cite{LimHL16}), and a real-world road network \cite{LiCH+05}. 

We consider the path version of \probh, where  the robot/UAV does not have to return to the root $r$ at the end; by Proposition~\ref{prop-tour-path}, all our results apply to  this path version. We also skip returning to the root at the end of each round: the solution goes directly from the last node of a  round   to the first node of the next round. (By triangle inequality, the modified  solution is at least as good as the solution from Algorithm~\ref{alg:k-rounds}.) We note that the experiments in  \cite{LimHL16} were also for the path version. 
{These experiments show that our algorithm, using a small number of rounds, performs well when compared with fully-adaptive algorithms, and also demonstrates the computational benefits of limited adaptivity.} 

All of our computations are run on a machine with a 3.0 GHz Intel Xeon Gold 6154 processor with $4$ cores: we assigned 28GB of memory to each experiment. The code is run using Python 3.10.4 and Gurobi 10.0.3.

\subsection{Instances}\label{sec:inst}

\begin{figure}[htbp]
\centering
  \includegraphics[width=0.4\linewidth]{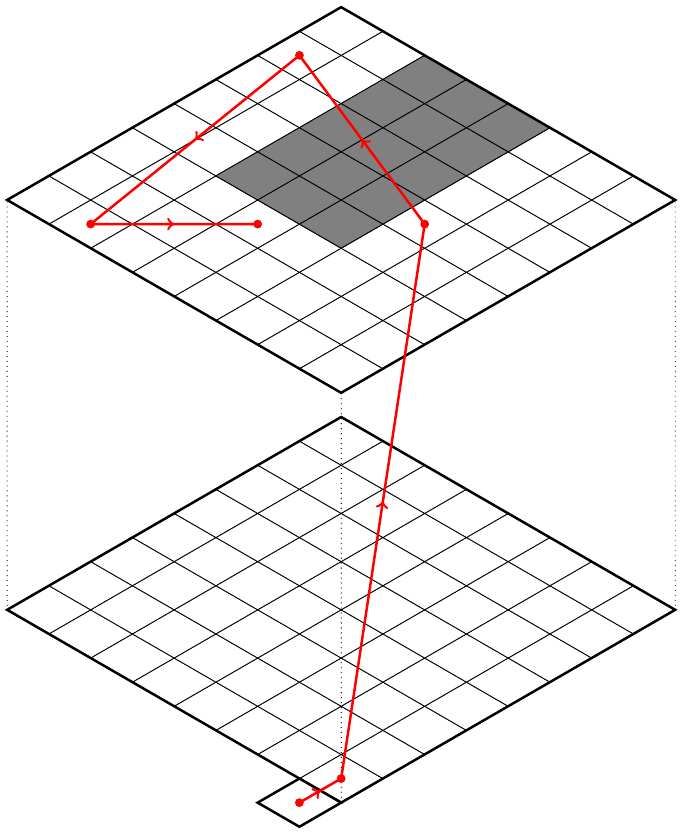}
  \caption{UAV search instance}
  \label{fig:grid}
\end{figure}

\paragraph{UAV Search} The UAV searches for a lost victim 
in an area modeled as an $N \times N$ grid (we set $N \in \{8, 9, 10\}$),
and must identify the location of the victim (see Figure~\ref{fig:grid}). 
We assume a uniform prior on the victim's initial location.
The UAV is equipped with two sensors, a long-range sensor and a short-range sensor, allowing the UAV to operate at different altitudes.
At the low altitude, the UAV can use the (more accurate) short-range sensor to
determine 
whether the current grid 
contains the victim.
The long-range sensor allows the UAV to
determine whether the victim is 
in the $3 \times 3$ grid around its current location.
We also allow for occlusion at high altitudes: using the long-range sensor from such locations yields no information (see Figure~\ref{fig:occl} for details).
We label the instances using $\mathtt{UAV}-N$ and add an identifier if
the instance has occlusions;
for e.g., the instance with $N=9$ and having occlusions is labeled $\mathtt{UAV}-9-\mathtt{OC}$.
We assume that the UAV must start and end at a root node $r$ that lies at the lower altitude outside the grid.
We set the cost of traveling between adjacent locations at the high (resp. low) altitude as $1$ (resp. $4$). The cost of moving between altitudes is $10$.

\noindent\textbf{California Road Network.} 
The road network dataset
describes connections between ``points of interests'' in California, and contains $21,047$ nodes and $21,691$ edges. 
We observe that a majority of the nodes in the network are degree-$2$ nodes; so we can combine such nodes and only keep nodes with degree $3$ or more. 
The resulting network has $1,365$ nodes and $1,991$ edges (see Figure~\ref{fig:cal-road-network}). 

\begin{figure}[htbp]
\centering
  \centering
  \includegraphics[width=0.6\linewidth]{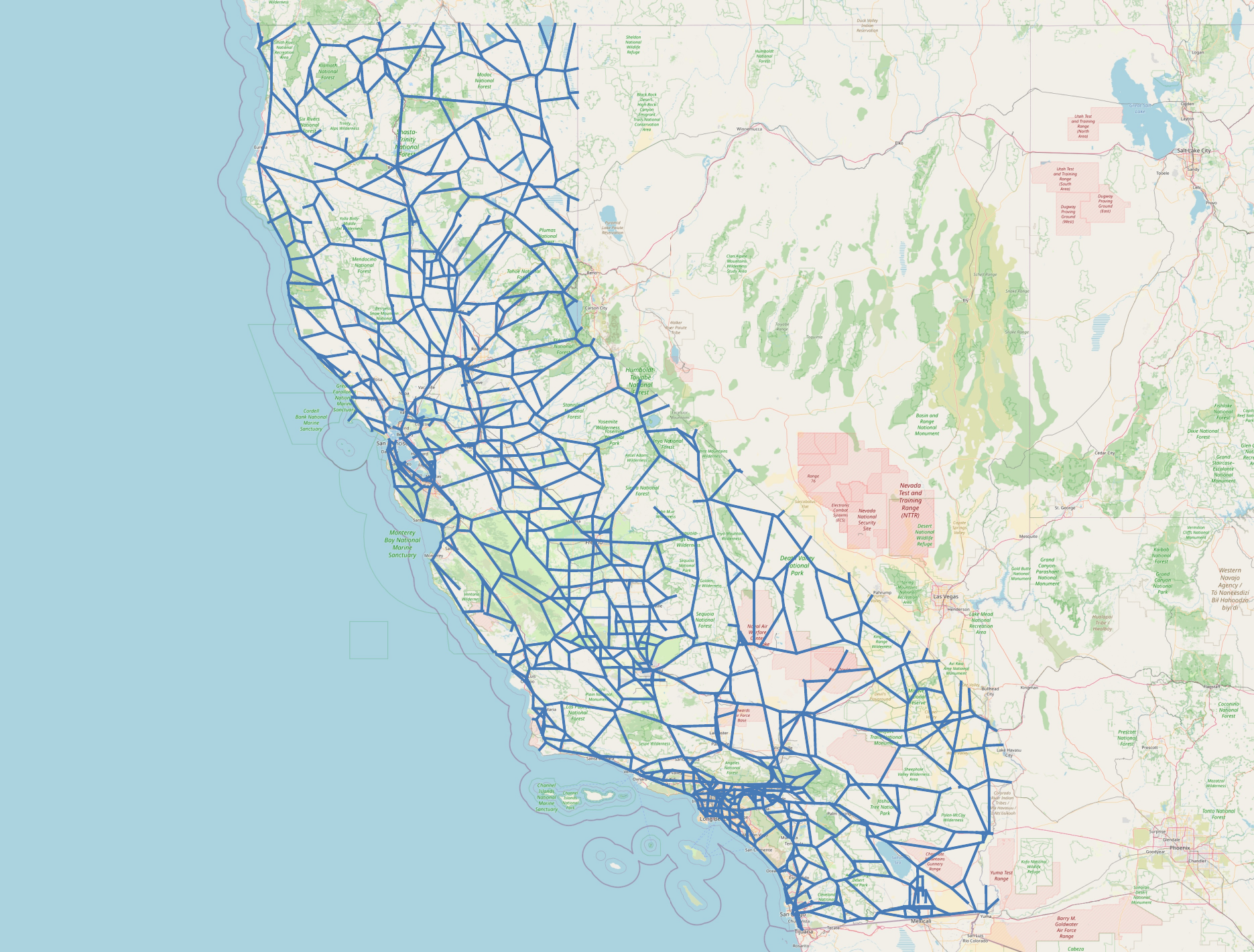}
\caption{California road network.}
\label{fig:cal-road-network}
\end{figure}

We use this network to generate 10 instances of \probh as follows. 
We set $m = 50$ and generate the scenarios using the independent cascade model (ICM).

The independent cascade model (ICM), which is commonly used in epidemic modeling and viral marketing,  
considers a diffusion process across a network. 
The model takes as
input a
directed graph $G(V, E)$, and a parameter $p$ 
that denotes the probability of a node influencing its neighbor. 
The process begins by marking the seed node, say $s$, as active. 
In each subsequent iteration, an active node
independently influences each of its neighbor
with probability $p$: this can happen only once.
If an active node $i$ fails to influence node $j$, we mark arc $(i,j)$ as blocked to prevent further activation attempts. The process terminates when no new activation occurs. 

For an instance, we generate a scenario $i$ as follows: we first sample a random seed node $s_i$ and mark it as active. We then simulate the influence spread using ICM. The set of nodes, say $S_i$, that are active upon termination constitute scenario $i$. Given scenario $i$, the robot receives a feedback of $1$ from node $v$ if $v \in S_i$, and $0$ otherwise. 
The robot begins near the node with the highest centrality measure, and the goal is to identify the underlying scenario with minimum expected cost.

\noindent \textbf{Choice of $p$.} We set 
$p\in \{0.6, 0.62, 0.64, 0.66, 0.68$\} to ensure that each scenario is nontrivial. When $p$ is less than $0.60$, the scenarios tend to contain only one node (since the degree of each node is not high); and setting $p$ greater than $0.7$ 
leads to 
some scenarios comprising the entire graph.
We present 
details regarding each instance
in Table~\ref{tab:sensing_stats}. 
We represent each instance as road-$p*100$-$m$(-$2$),
where $p$ denotes the probability
with which an active node can influence its neighbor,
and $m$ denote the number of scenarios. 
The suffix $(-2)$ is added since we generate two instances for every value of $p$.
We note that as $p$ increases, both the number of scenarios each node is contained in, and the number of 
nodes in each scenario 
increases.

\begin{table}[H]
\centering
\caption{Statistics of Sensing Matrix}
\label{tab:sensing_stats}
\begin{tabular}{lcc}
\toprule
 Instance & Avg. \# scenarios  & Avg. \# nodes \\
type &  per node &  per scenario  \\
\midrule
road-60-50 & 3.04 & 83.00 \\
road-60-50-2 & 3.10 & 84.52 \\
road-62-50 & 3.36 & 91.68 \\
road-62-50-2 & 4.01 & 109.36 \\
road-64-50 & 4.23 & 115.50 \\
road-64-50-2 & 6.64 & 181.40 \\
road-66-50 & 6.33 & 172.82 \\
road-66-50-2 & 6.13 & 167.34 \\
road-68-50 & 10.37 & 283.10 \\
road-68-50-2 & 10.83 & 295.76 \\
\bottomrule
\end{tabular}
\end{table}

\subsection{Results}
\noindent \textbf{UAV Search.} 
We run our $k$-round algorithm for $k \in \{1,\dots, 10\}\cup \{\mathtt{inf}\}$, where $\mathtt{inf}$ denotes the fully adaptive algorithm. 
For each instance, 
we record and report the cost to identify each scenario.
We normalize the cost achieved by the ${k}$-round algorithm against the adaptive algorithm (per scenario), yielding an average relative cost (ARC), given by the following equation:
\begin{equation}\label{eq:arc}
    \text{ARC} = \E_{i}\left[\frac{k\text{-round}(i) - \mathtt{ADAP}(i)}{\mathtt{ADAP}(i)}\right] \times 100\%
\end{equation}
where $k$-round$(i)$ (resp. $\mathtt{ADAP}(i)$) denotes the cost incurred by the $k$-round (resp. fully adaptive) algorithm for scenario $i$.
We also report the average  CPU time (over scenarios) taken by the $k$-round algorithm. 
We present these results in Tables \ref{tab:uav_cost}, 
\ref{tab:uav_arc} and \ref{tab:uav-time} respectively.
We observe (see Figures~\ref{fig:uav_cost}-\ref{fig:uav_time})
that, with an increase in the number of rounds, the cost (ARC) typically decreases and computation time increases. 
However, the trend of ARC sometimes shows an increased cost with more rounds: this can be attributed to  the randomness in our algorithm (Theorem~\ref{thm:k-round-id}), which is due to use of  probabilistic tree embedding.
One notable result  is for the UAV instance $8-\mathtt{OC}$, which was the only instance tested in \cite{LimHL16}:  even with 2 rounds, our solution cost 
is about $30\%$ less
than the  fully-adaptive solution found  in  \cite{LimHL16} (they reported a cost of  $83.6$ for their fully-adaptive algorithm). 
We observe that using $2$ rounds of adaptivity, the cost  is on average within  $12\%$  of the fully adaptive cost. Moreover, the $2$-round algorithm is $20$ times faster than the fully adaptive one (averaging over instances and scenarios).

\begin{figure}[htbp]
\centering
\begin{subfigure}{.3\textwidth}
  \centering
  \includegraphics[width=0.6\linewidth]{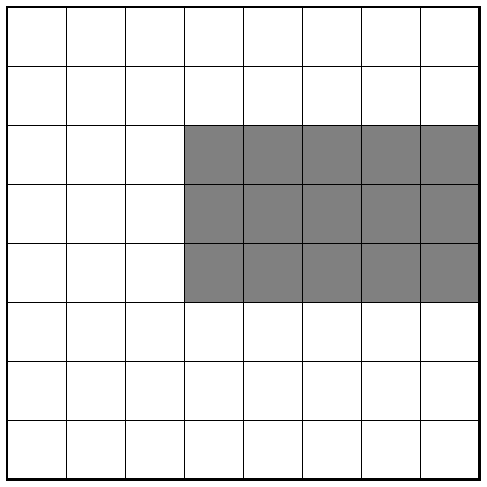}
  \caption{$\mathtt{UAV}-8-\mathtt{OC}$}
  \label{fig:eight_occ}
\end{subfigure}%
\begin{subfigure}{.3\textwidth}
  \centering
  \includegraphics[width=0.6\linewidth]{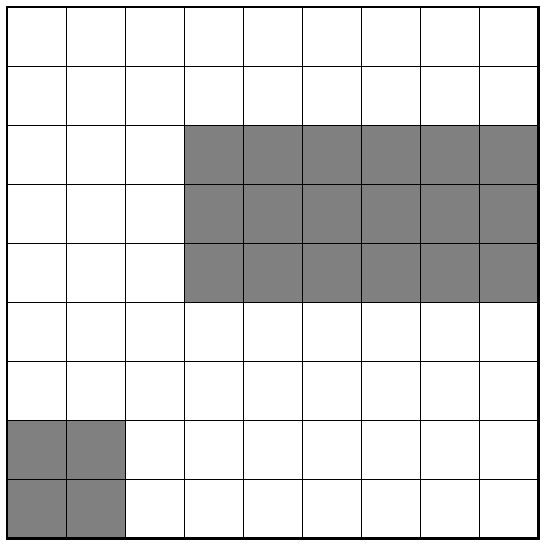}
  \caption{$\mathtt{UAV}-9-\mathtt{OC}$}
  \label{fig:nine_occ}
\end{subfigure}%
\begin{subfigure}{0.3\textwidth}
  \centering
  \includegraphics[width=0.6\linewidth]{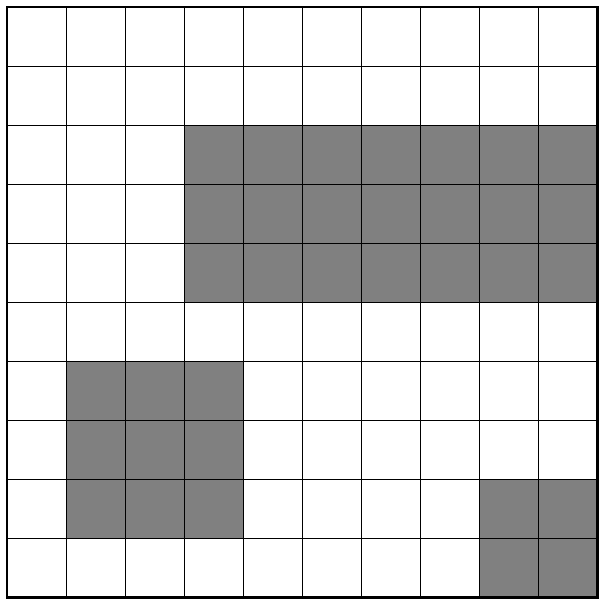}
  \caption{$\mathtt{UAV}-10-\mathtt{OC}$}
  \label{fig:ten_occ}
\end{subfigure}
\caption{Occlusions used for UAV search instances.}
\label{fig:occl}
\end{figure}

\begin{table}[H]
\centering
\caption{Average cost for UAV search}
\label{tab:uav_cost}
\begin{tabular}{lrrrrrrrr}
\toprule
 & \multicolumn{8}{c}{Path Cost} \\
\# Rounds & 1 & 2 & 3 & 4 & 5 & 6 & 7 & inf \\
Instance Type &  &  &  &  &  &  &  &  \\
\midrule
uav-8 & 80.92 & 58.41 & 56.66 & 54.69 & 57.94 & 59.67 & 58.34 & 51.34 \\
uav-8-OC & 108.25 & 60.64 & 56.80 & 54.39 & 57.42 & 54.53 & 56.69 & 57.94 \\
uav-9 & 132.65 & 96.83 & 88.33 & 87.72 & 86.75 & 88.02 & 93.25 & 87.00 \\
uav-9-OC & 155.95 & 86.32 & 91.23 & 88.77 & 92.31 & 93.31 & 90.31 & 92.33 \\
uav-10 & 157.31 & 91.10 & 91.10 & 90.12 & 89.22 & 86.71 & 88.24 & 86.59 \\
uav-10-OC & 209.79 & 115.38 & 113.60 & 111.90 & 110.27 & 109.06 & 110.23 & 104.86 \\
\bottomrule
\end{tabular}
\end{table}

\begin{table}[H]
\centering
\caption{Average relative cost per scenario (vs. fully adaptive) for UAV search}
\label{tab:uav_arc}
\begin{tabular}{lrrrrrrrr}
\toprule
 & \multicolumn{8}{c}{Relative Cost Per Scenario (vs. Fully Adaptive)} \\
\# Rounds & 1 & 2 & 3 & 4 & 5 & 6 & 7 & inf \\
Instance Type &  &  &  &  &  &  &  &  \\
\midrule
uav-8 & 66.70 & 16.92 & 13.02 & 9.56 & 15.90 & 19.69 & 15.24 & 0.00 \\
uav-8-OC & 105.29 & 12.14 & 10.92 & 6.08 & 10.10 & 2.97 & 7.96 & 0.00 \\
uav-9 & 81.67 & 27.68 & 6.13 & 5.34 & 4.25 & 5.92 & 11.95 & 0.00 \\
uav-9-OC & 100.91 & -0.20 & 3.97 & 0.94 & 4.63 & 4.83 & 3.84 & 0.00 \\
uav-10 & 97.80 & 6.86 & 6.38 & 5.92 & 5.10 & 3.53 & 5.29 & 0.00 \\
uav-10-OC & 135.30 & 14.37 & 12.63 & 11.07 & 9.59 & 8.40 & 9.51 & 0.00 \\
\bottomrule
\end{tabular}
\end{table}

\begin{table}[H]
\centering
\caption{Average planning time for UAV  search across all scenarios}
\label{tab:uav-time}
\begin{tabular}{lrrrrrrrr}
\toprule
 & \multicolumn{8}{c}{Planning Time (s)} \\
\# Rounds & 1 & 2 & 3 & 4 & 5 & 6 & 7 & inf \\
Instance Type &  &  &  &  &  &  &  &  \\
\midrule
uav-8 & 0.49 & 0.47 & 0.49 & 0.47 & 0.87 & 0.94 & 0.95 & 5.37 \\
uav-8-OC & 0.44 & 0.63 & 0.50 & 0.49 & 0.95 & 0.77 & 0.74 & 7.86 \\
uav-9 & 0.73 & 0.76 & 0.79 & 0.80 & 1.02 & 2.45 & 3.90 & 21.94 \\
uav-9-OC & 0.68 & 0.58 & 0.71 & 0.97 & 1.22 & 1.04 & 3.35 & 14.81 \\
uav-10 & 0.83 & 0.63 & 0.65 & 0.83 & 0.89 & 5.62 & 5.69 & 5.27 \\
uav-10-OC & 0.73 & 0.71 & 0.80 & 0.96 & 2.28 & 1.96 & 6.29 & 19.74 \\
\bottomrule
\end{tabular}
\end{table}

\begin{figure}[H]
    \centering
    \includegraphics[width=0.6\linewidth]{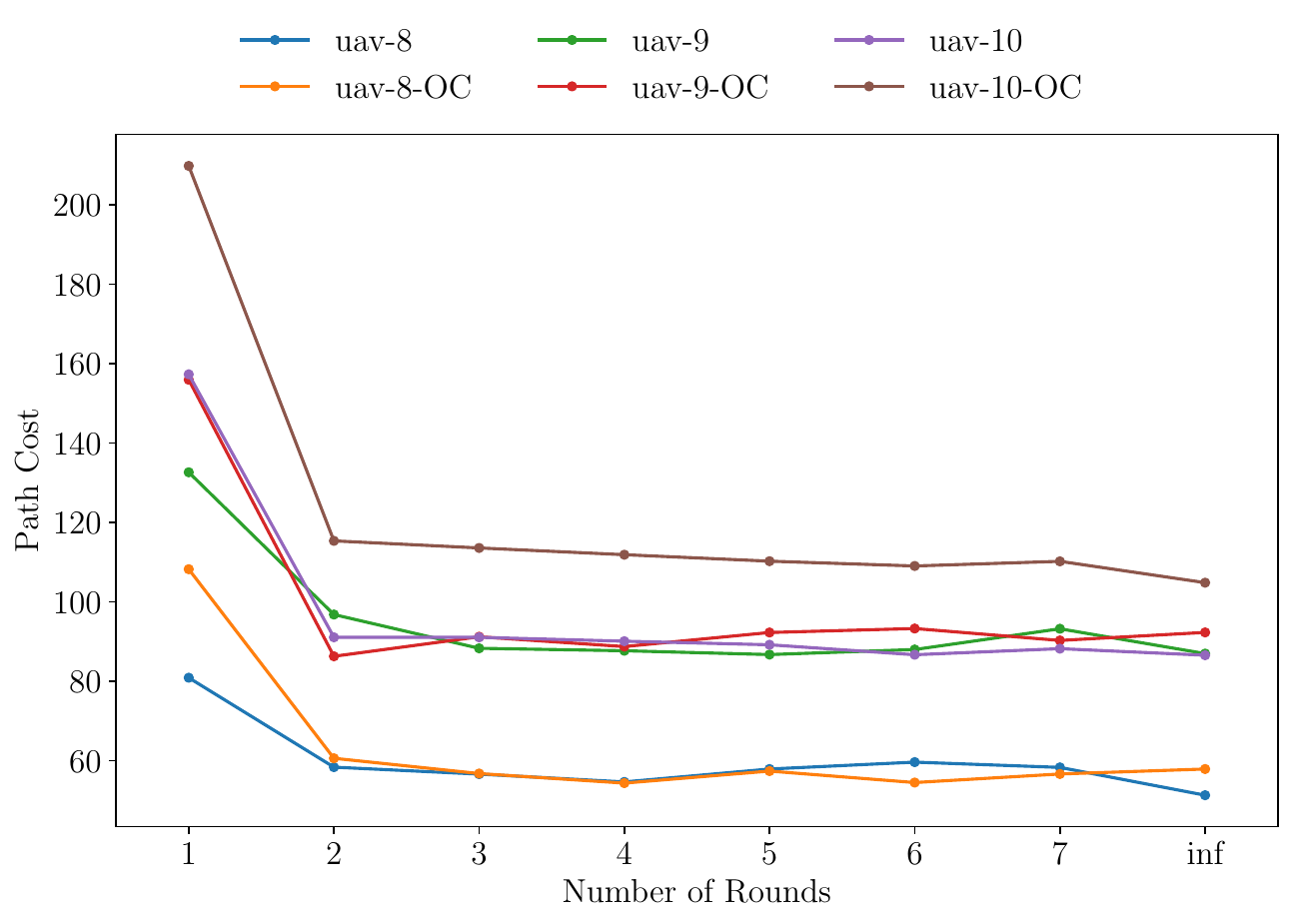}
    \caption{Average cost v/s \#rounds for UAV search}
    \label{fig:uav_cost}
\end{figure}

\begin{figure}[H]
    \centering
    \includegraphics[width=0.6\linewidth]{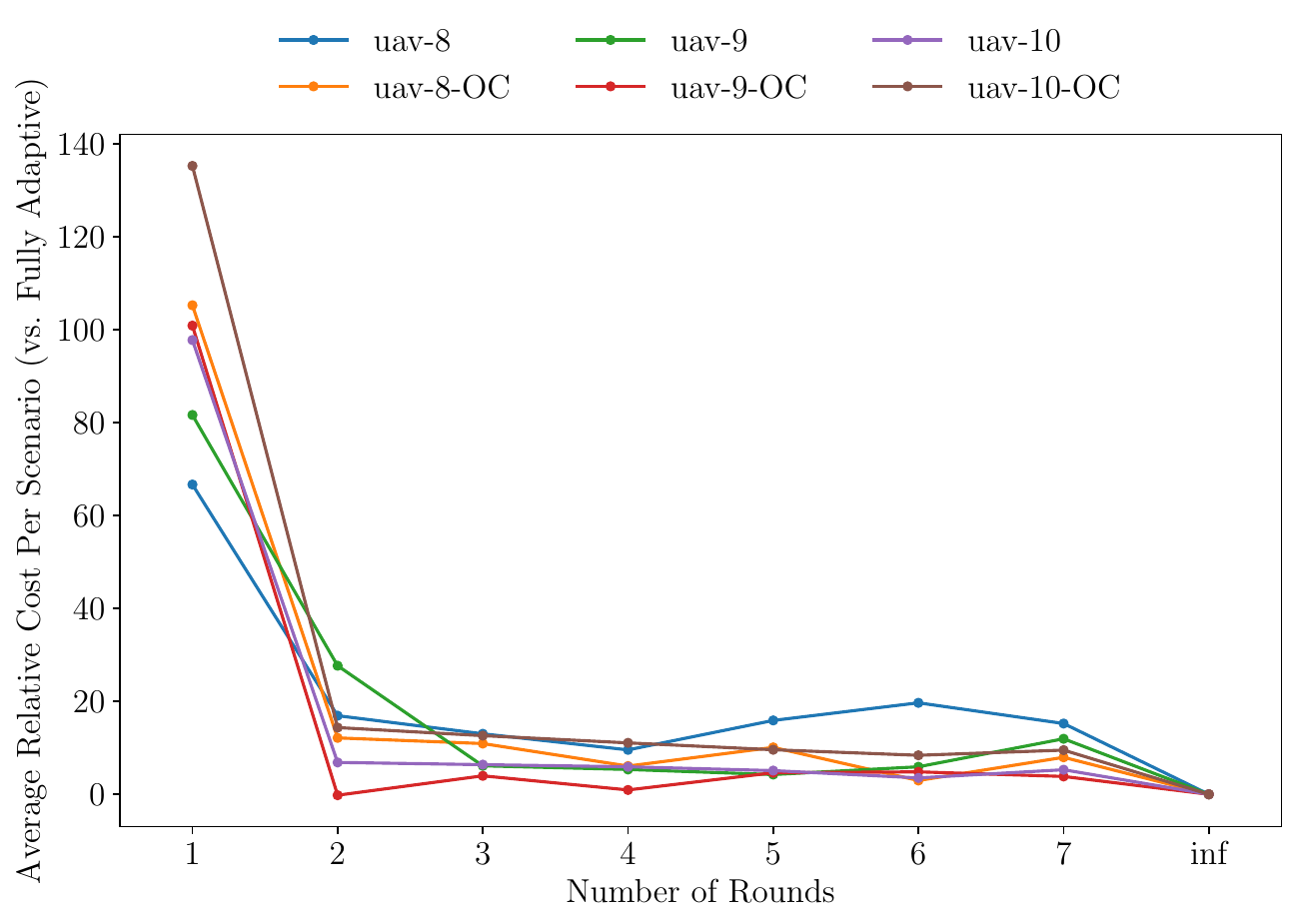}
    \caption{Average relative cost per scenario v/s \#rounds for UAV search}
    \label{fig:uav_arc}
\end{figure}

\begin{figure}[H]
    \centering
    \includegraphics[width=0.6\linewidth]{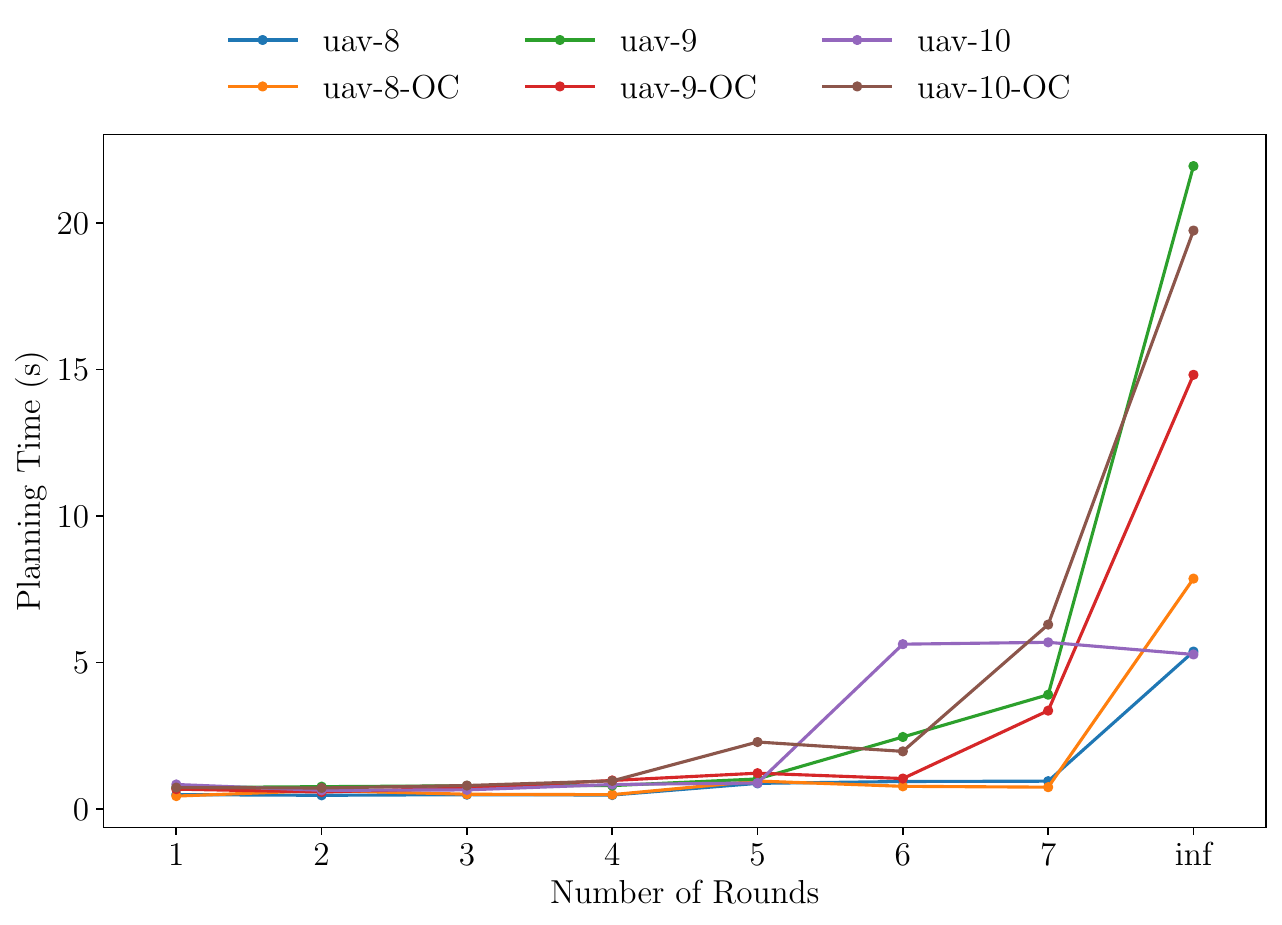}
    \caption{Average time v/s  \#rounds for UAV search}
    \label{fig:uav_time}
\end{figure}

We now present the computational results for instances generated using a California road network dataset \cite{LiCH+05}.

\noindent \textbf{California Road Network.}
As in the case of UAV search, we run our $k$-round algorithm for $k \in \{1,\dots, 10\}\cup \{\mathtt{inf}\}$, where $\mathtt{inf}$ denotes the fully adaptive algorithm. 
For each instance, 
we record and report the cost to identify each scenario.
We normalize the cost achieved by the ${k}$-round algorithm against the adaptive algorithm (per scenario), yielding an average relative cost (ARC), given by \eqref{eq:arc}.
We also report the average  CPU time (over scenarios) taken by the $k$-round algorithm. 
We present these results in Tables~\ref{tab:road_cost}-\ref{tab:road_time}.
We observe similar trends as observed for UAV search instances: 
with an increase in the number of rounds, the cost (ARC) typically decreases and computation time increases (see Figures~\ref{fig:road_cost}-\ref{fig:road_time}). 
We note that using $2$ rounds of adaptivity, the cost  is on average within $50\%$ of the fully adaptive cost, while being 
significantly ($\approx 15$ times) faster than than the fully adaptive algorithm.
We also note that the fully adaptive algorithm for instance $\mathtt{road-66-50-2}$ ran out of memory in the original experiment, and required the memory to be raised to $56\text{GB}$.

\begin{table}[H]
\centering
\caption{Average Cost for Road Networks (in km)}
\label{tab:road_cost}
\begin{tabular}{lrrrrrrrrrrr}
\toprule
 & \multicolumn{11}{c}{Path Cost} \\
\# Rounds & 1 & 2 & 3 & 4 & 5 & 6 & 7 & 8 & 9 & 10 & inf \\
Instance Type &  &  &  &  &  &  &  &  &  &  &  \\
\midrule
road-60-50 & 4735 & 3444 & 2941 & 3130 & 2878 & 3215 & 3374 & 3387 & 3668 & 3092 & 3423 \\
road-60-50-2 & 3992 & 2835 & 3028 & 2826 & 3070 & 3156 & 3129 & 3235 & 3469 & 2985 & 3714 \\
road-62-50 & 4833 & 4324 & 4139 & 4221 & 4438 & 4565 & 4438 & 4333 & 4200 & 4180 & 3618 \\
road-62-50-2 & 3932 & 2726 & 2916 & 2949 & 3565 & 3463 & 3846 & 3399 & 4215 & 3484 & 3590 \\
road-64-50 & 4412 & 4406 & 4288 & 4135 & 4612 & 4307 & 4571 & 4148 & 4168 & 4208 & 3976 \\
road-64-50-2 & 3070 & 3106 & 3169 & 3075 & 3316 & 3799 & 3580 & 3878 & 3049 & 3540 & 2915 \\
road-66-50 & 3909 & 2901 & 3001 & 2853 & 2966 & 2877 & 3114 & 2347 & 2537 & 2482 & 2566 \\
road-66-50-2 & 4235 & 2391 & 2685 & 2946 & 2641 & 2897 & 3449 & 2882 & 3122 & 3066 & 3257 \\
road-68-50 & 4285 & 2681 & 2589 & 2823 & 2531 & 2164 & 2651 & 2166 & 1987 & 2256 & 2458 \\
road-68-50-2 & 3913 & 2606 & 2568 & 2140 & 1898 & 2322 & 2161 & 2470 & 2720 & 2671 & 2071 \\
\bottomrule
\end{tabular}
\end{table}

\begin{table}[H]
\centering
\caption{Average Relative Cost Per Scenario (vs. Fully Adaptive) for Road Networks}
\label{tab:road_arc}
\begin{tabular}{lrrrrrrrrrrr}
\toprule
 & \multicolumn{11}{c}{Relative Cost Per Scenario (vs. Fully Adaptive)} \\
\# Rounds & 1 & 2 & 3 & 4 & 5 & 6 & 7 & 8 & 9 & 10 & inf \\
Instance Type &  &  &  &  &  &  &  &  &  &  &  \\
\midrule
road-60-50 & 83.86 & 24.36 & 8.80 & 33.63 & 11.31 & 37.46 & 51.51 & 36.63 & 43.83 & 18.58 & 0.00 \\
road-60-50-2 & 51.54 & 54.14 & 26.49 & 13.15 & 29.73 & 28.78 & 25.18 & 29.83 & 25.55 & 21.50 & 0.00 \\
road-62-50 & 110.30 & 81.56 & 97.93 & 80.09 & 104.56 & 96.14 & 97.55 & 86.63 & 81.51 & 75.64 & 0.00 \\
road-62-50-2 & 51.62 & 14.94 & 17.05 & 25.55 & 18.38 & 27.68 & 46.37 & 15.56 & 29.00 & 25.46 & 0.00 \\
road-64-50 & 59.03 & 49.29 & 67.36 & 46.45 & 109.37 & 61.42 & 51.80 & 54.77 & 74.68 & 26.16 & 0.00 \\
road-64-50-2 & 55.77 & 81.72 & 97.88 & 44.55 & 85.76 & 82.43 & 60.21 & 130.09 & 34.99 & 45.92 & 0.00 \\
road-66-50 & 197.39 & 70.03 & 80.56 & 68.84 & 65.76 & 72.69 & 90.66 & 49.59 & 108.24 & 17.90 & 0.00 \\
road-66-50-2 & 105.99 & 34.27 & 64.21 & 50.34 & 30.49 & 16.60 & 58.93 & 45.78 & 81.00 & 60.45 & 0.00 \\
road-68-50 & 125.46 & 33.59 & 32.77 & 35.27 & 58.29 & 30.68 & 18.78 & 10.97 & 1.78 & 22.78 & 0.00 \\
road-68-50-2 & 155.13 & 54.58 & 45.54 & 35.69 & 6.90 & 27.25 & 33.16 & 28.30 & 50.32 & 37.84 & 0.00 \\
\bottomrule
\end{tabular}
\end{table}

\begin{table}[H]
\centering
\caption{Average Planning Time of Road Networks Across All Scenarios}
\label{tab:road_time}
\begin{tabular}{lrrrrrrrrrrr}
\toprule
 & \multicolumn{11}{c}{Planning Time (s)} \\
\# Rounds & 1 & 2 & 3 & 4 & 5 & 6 & 7 & 8 & 9 & 10 & inf \\
Instance Type &  &  &  &  &  &  &  &  &  &  &  \\
\midrule
road-60-50 & 4.66 & 7.09 & 8.31 & 12.35 & 15.88 & 18.53 & 25.38 & 31.25 & 36.29 & 32.92 & 119.08 \\
road-60-50-2 & 4.98 & 6.69 & 9.20 & 12.03 & 15.52 & 20.59 & 23.65 & 23.24 & 36.75 & 42.14 & 143.84 \\
road-62-50 & 6.22 & 6.56 & 7.89 & 11.70 & 15.35 & 21.60 & 29.77 & 33.78 & 29.94 & 37.69 & 123.86 \\
road-62-50-2 & 5.44 & 6.53 & 10.61 & 11.95 & 15.27 & 19.53 & 28.94 & 34.20 & 41.20 & 40.63 & 94.11 \\
road-64-50 & 4.81 & 6.04 & 8.81 & 11.65 & 14.17 & 22.38 & 25.58 & 24.60 & 43.90 & 50.98 & 112.40 \\
road-64-50-2 & 6.49 & 7.14 & 10.80 & 13.06 & 15.51 & 30.87 & 28.48 & 27.34 & 46.05 & 54.69 & 106.31 \\
road-66-50 & 5.88 & 7.04 & 9.45 & 15.08 & 16.64 & 23.08 & 25.49 & 21.40 & 36.91 & 36.65 & 109.96 \\
road-66-50-2 & 7.48 & 10.46 & 11.76 & 19.93 & 15.04 & 23.37 & 30.84 & 40.69 & 39.50 & 39.43 & 217.21 \\
road-68-50 & 6.99 & 9.97 & 14.07 & 22.13 & 29.11 & 38.28 & 41.90 & 38.41 & 52.52 & 55.36 & 84.31 \\
road-68-50-2 & 6.30 & 8.74 & 12.28 & 15.89 & 15.27 & 16.93 & 18.70 & 41.87 & 51.14 & 54.97 & 112.10 \\
\bottomrule
\end{tabular}
\end{table}

\begin{figure}[H]
    \centering
    \includegraphics[width=0.6\linewidth]{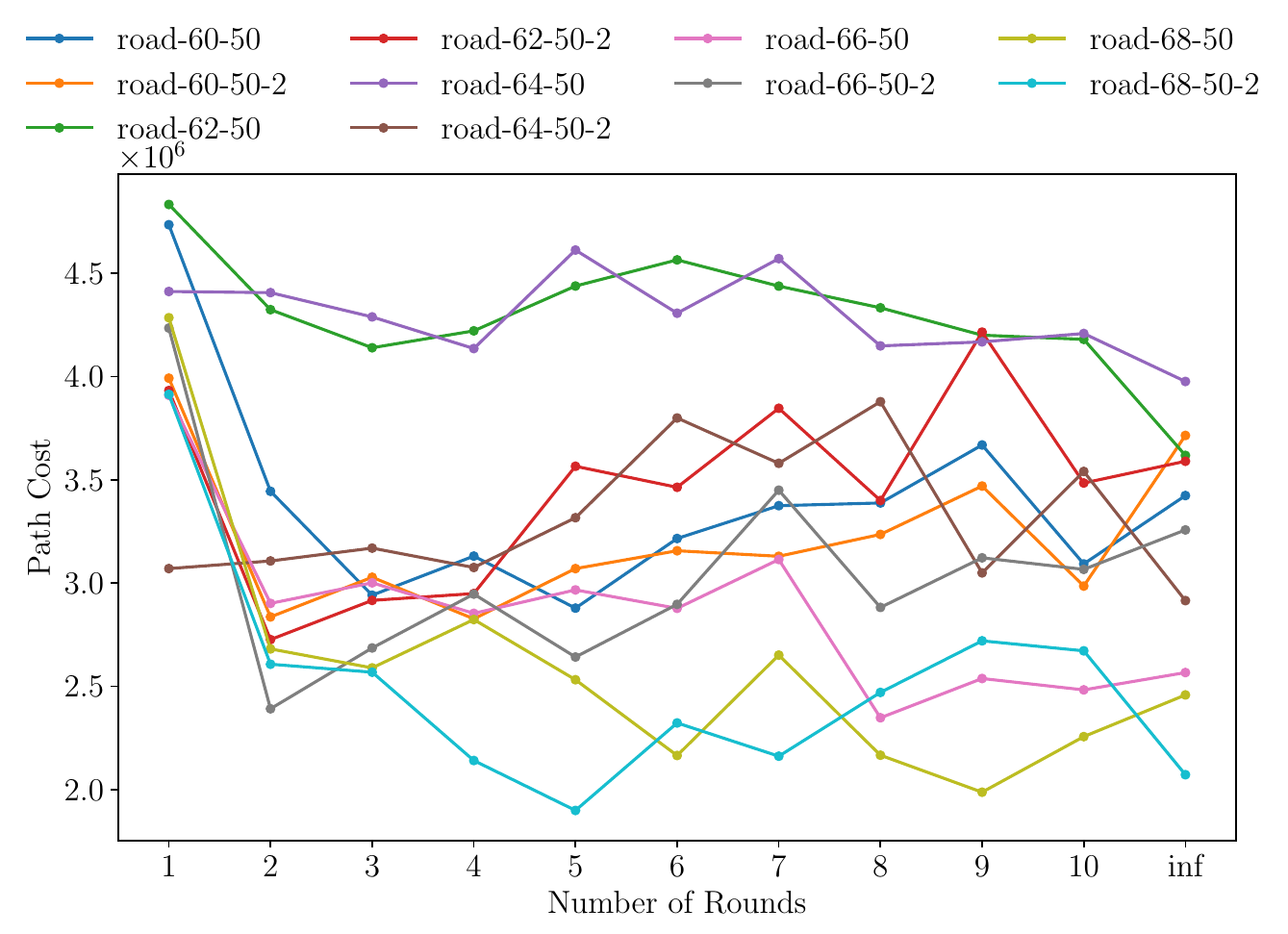}
    \caption{Average cost v/s \#rounds for California road network}
    \label{fig:road_cost}
\end{figure}

\begin{figure}[H]
    \centering
    \includegraphics[width=0.6\linewidth]{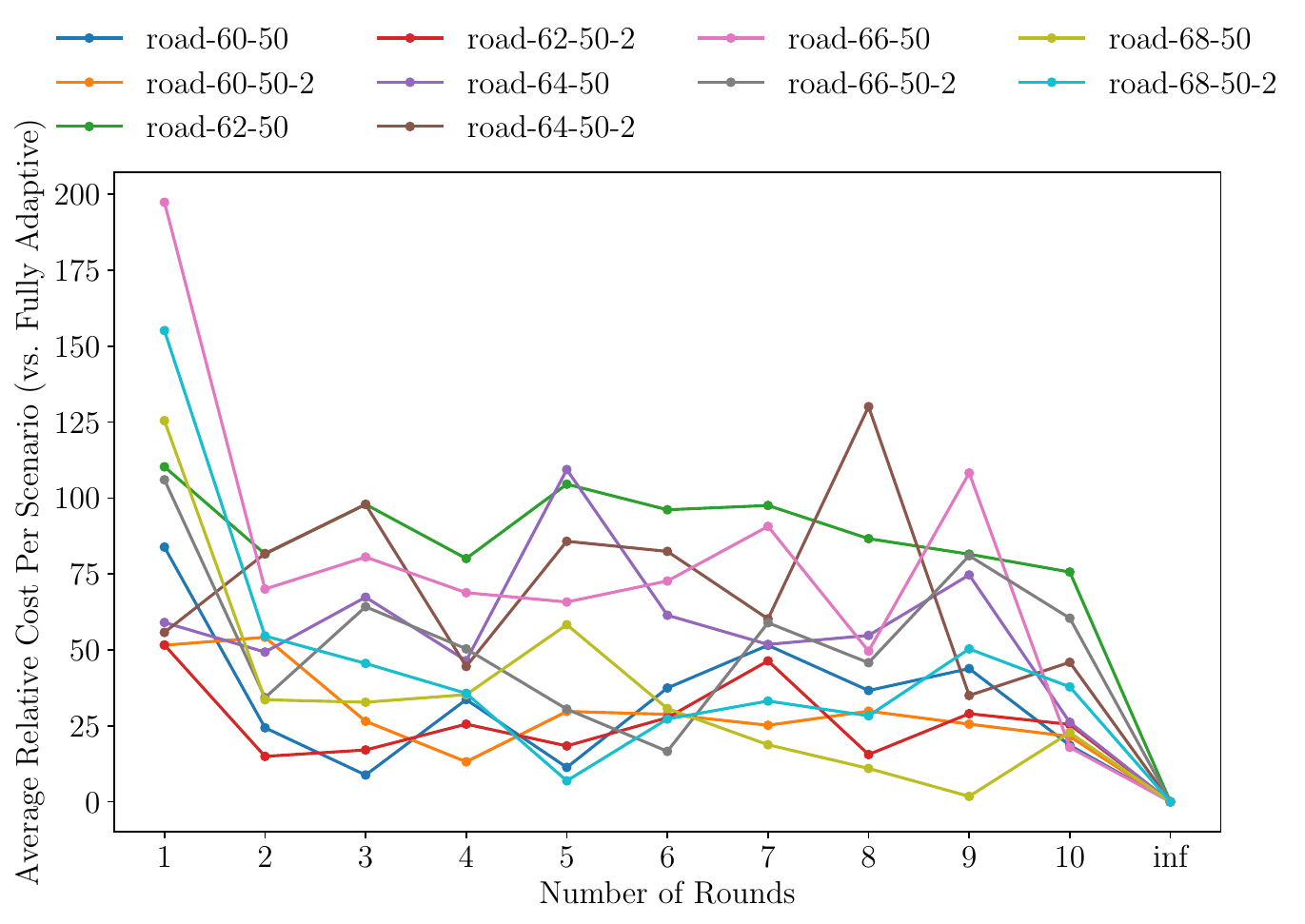}
    \caption{Average relative cost per scenario v/s \#rounds for California road network}
    \label{fig:road_arc}
\end{figure}

\begin{figure}[H]
    \centering
    \includegraphics[width=0.6\linewidth]{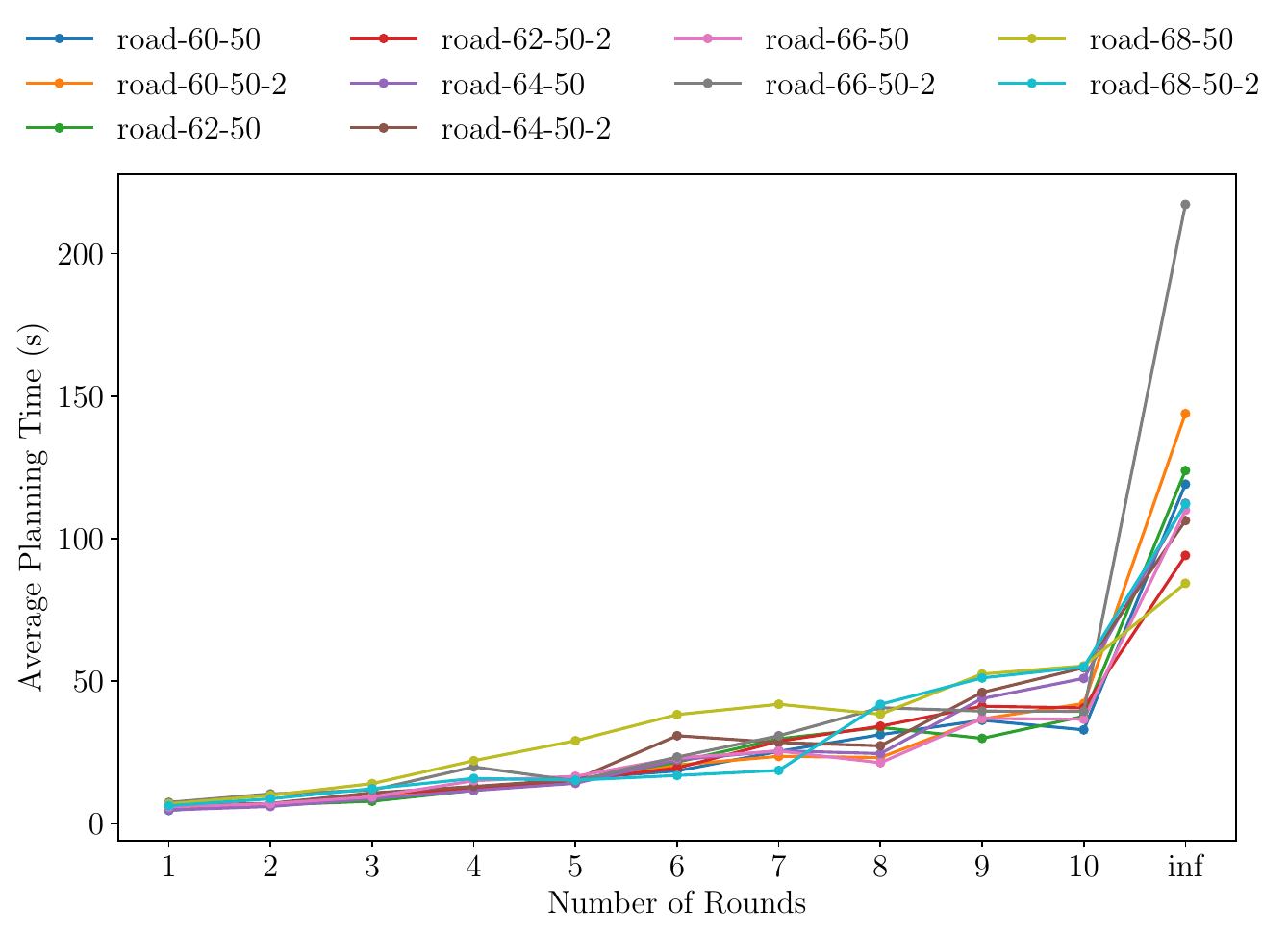}
    \caption{Average time v/s \#rounds for California road network}
    \label{fig:road_time}
\end{figure}

\vfill
\vspace{-2mm}
\section{Conclusion}
\vspace{-2mm}
In this paper, we design an algorithm for the informative path planning
problem parameterized by the number $k$ of adaptive rounds, and prove a smooth trade-off between $k$ and
the solution quality. 

Our computational experiments corroborate our theory, showing that a  few rounds of adaptivity suffice to get solutions comparable to fully adaptive ones, 
while providing a significant benefit in computational time. 
We leave open the question 
of designing algorithms under  uncertainty of  measurements:  can we design algorithms with limited adaptivity when  measurements may be imprecise or even incorrect with some  probability?

{
\bibliographystyle{alpha}
\bibliography{ref}
}
\newpage
\appendix
\section{Tighter Approximation Using More Rounds}\label{app:better-appx}

We now present a better approximation for the informative path planning problem. 
We note that a similar result for the scenario submodular cover problem appeared in \cite{GGN21}. Their result directly applies to \prob: we include the details here for completeness.
The main result of this section is as follows.

\begin{theorem}\label{thm:2k-round-ipp}
For any integer $k \geq 1$ and constant $\epsilon > 0$, there is a $2k$-round adaptive algorithm for the informative path planning problem 
with cost at most $\bigOh\left(\log^{2+\epsilon}(n)\cdot m^{1/k}\cdot \log(mQ)\right)$ times the cost of an optimal adaptive algorithm.
\end{theorem}
Combining Theorem~\ref{thm:2k-round-ipp}, and the result for the ratio group Steiner tree problem (Theorem~\ref{thm:grp-st-ratio}), we obtain the following.

\begin{corollary}\label{cor:2k-round-id}
For any integer $k \geq 1$, there is a $2k$-round adaptive algorithm for the 
hypothesis identification
problem with cost 
at most $\bigOh\left(\log^{2}(n) \cdot m^{1/k}\cdot \log m)\right)$ times the cost of an optimal
adaptive algorithm.
\end{corollary}

The main idea to prove Theorem~\ref{thm:2k-round-ipp} is to introduce the following \emph{more general} partial cover version
of \prob.
An instance of this partial cover version of \prob is the same as an instance of \prob with two additional parameters $\delta, \eta \in (0, 1]$. 
The goal is to visit a set of locations $T$ that realize to $\psi(T)\in \Psi$ such that either 
(i) number of compatible scenarios  $|\{\scn\in M: \psi(T) \sse \psi_{\scn}\}|<\delta m$, or (ii) the function $f$ is partially covered, i.e., $f(\psi(T)) > Q \cdot (1 - \eta)$. Note that $\eta = 1/Q$ recovers the original partial cover version (since $f$ is integral).
We refer to this problem as $\mathtt{GPC}$-\prob, and prove the following.

\begin{theorem}\label{thm:better-partial-cover}
There is a non-adaptive algorithm for $\mathtt{GPC}$-\prob 
with expected cost 
$\bigOh \left(\frac{\rho}{\delta}\log\left(\frac{1}{\delta \eta}\right)\right)$ times the cost of the optimal adaptive solution for \prob,
where $\rho$ is the best  approximation guarantee for ratio submodular orienteering.
\end{theorem}
\begin{proof}
The non-adaptive algorithm for $\mathtt{GPC}$-\prob is identical to $\mathtt{PCA}$
(Algorithm~\ref{alg:partial-cover}) with one crucial difference. 
We modify the definition of $\mathcal{Z}$ in Definition~\ref{defn:partn-scn} to 
exclude parts where the coverage exceeds $Q \cdot (1 - \eta)$. 
We now formally define $\mathcal{Z} = \{Y \in \H(S) : |Y| \geq \delta |M| \text{ and } f(\psi_Y(S)) \leq Q \cdot (1 - \eta)\}$.
The scenarios in $\mathcal{Z}$ correspond to those 
for which the stopping rule does not apply, and this captures the 
feasibility criteria for $\mathtt{GPC}$-\prob.
The algorithm stops when either (i) the number of compatible scenarios is less than $\delta m$, or (ii) $f(\psi(T)) > Q \cdot (1 - \eta)$
The rest of the algorithm (score function, tour generation, and probing remains the same).

The analysis 
is similar to that of
Theorem~\ref{thm:partial-cover-thm}. 
We prove Lemma~\ref{lem:key-lemma}.
The lower bound remains the same; see \eqref{eq:lower-bound-score}. For the upper bound \eqref{eq:sum-gain-i-ub}, the analysis for the first term
in remains the same, but its second term can be bounded as:
\begin{equation}\notag
      \mathlarger{\sum}_{p=h}^{\ell} \left( \frac{f\left(\psi_{\scn}(\Pi_p) \cup \cup_{q=1}^{p-1}\psi_{\scn}(\Pi_q) \right) - f\left(\cup_{q=1}^{p-1}\psi_{\scn}(\Pi_q)\right)}{Q - f\left(\cup_{q=1}^{p-1}\psi_{\scn}(\Pi_q)\right)}\right) \leq \sum_{p=\eta Q}^{Q}\frac{1}{p} \leq \log\left(\frac{1}{\eta}\right).
\end{equation}
Thus, $\sum_{j \geq i} G_j \leq \log\left(\frac{1}{\delta \eta}\right) \cdot a(i)$. This proves Lemma~\ref{lem:key-lemma} with $L = \log\left(\frac{1}{\delta \eta}\right)$, and thus proves the theorem. \end{proof}

We end this section with a proof of Theorem~\ref{thm:2k-round-ipp}.

\begin{proof}[Proof of Theorem~\ref{thm:2k-round-ipp}.]
We can iteratively use the non-adaptive algorithm for $\mathtt{GPC}$-\prob to get a $2k$-round algorithm for \prob. 
We set $\delta = m^{-1/k}$ and $\eta = Q^{-1/k}$ in each round.  We note that 
after each round, either (i) the number of compatible scenarios drop by a factor of $m^{1/k}$, or (ii) the remaining coverage value drops by a factor of $Q^{1/k}$. Furthermore, note that (i) can occur at most $k$ times, and by the integrality of $f$, if (ii) occurs at least $k$ times, the $f$ is covered. Thus, the above procedure terminates in at most $2k$ rounds.
\paragraph{Expected Cost.} 
We define the 
{\em state} 
of a round as a list of all previously visited locations and their corresponding observations.
For the analysis, we view the iterative algorithm  as  a $2k$ depth tree, where the nodes at depth $\ell$ are the   states  in round $\ell$ and the branches out of each node  represent the observed realizations in that round. 
Let the set $\Omega_{\ell}$ denote all the states in round $\ell$, for $\ell = 1, \ldots, 2k$.
Using 
Theorem~\ref{thm:better-partial-cover}, for any $\omega\in \Omega_\ell$
the expected cost in round $\ell$
is at most $\bigOh\left(\frac{1}{k}\cdot m^{1/k}\cdot \log(mQ)\right) \cdot \OPT(\omega)$
where
$\OPT(\omega)$ denotes the expected cost of the optimal adaptive policy conditioned on the realizations in $\omega$.
Hence, the (unconditional) expected cost in round $\ell$ is at most $\bigOh\left(\frac{1}{k}\cdot m^{1/k}\cdot \log(mQ)\right) \cdot \sum_{\omega\in \Omega_{\ell}} \Pr[\omega]\cdot\OPT(\omega) =\bigOh\left(\frac{1}{k}\cdot m^{1/k}\cdot \log(mQ)\right)\cdot \OPT$ 
where we used that $\Omega_{\ell}$ is a partition of all outcomes. So  the  expected cost of the $2k$-round algorithm is at most $2k\cdot \bigOh \left(\frac{1}{k}\cdot m^{1/k}\cdot \log(mQ)\right)\cdot \OPT$ as claimed. 

\end{proof}
\end{document}